\documentclass[12pt,journal,comsoc,onecolumn,12pt,draftcls,twoside]{IEEEtran}

\usepackage[T1]{fontenc}
\usepackage{cite}
\usepackage{lipsum}
\usepackage{graphicx}
\usepackage{amssymb}
\usepackage{amsmath, nccmath}
\usepackage{amsthm}
\usepackage{mathtools}
\usepackage{cuted}
\usepackage{multicol}
\usepackage{color}
\usepackage{xcolor}
\usepackage{subfigure}
\usepackage{soul}

\newtheorem{proposition}{Proposition}
\usepackage{relsize} 

\newcommand{\expect}{\textsf{E}}
\newcommand{\transpose}{\textsf{T}}
\newcommand{\hermitian}{\textsf{H}}
\newcommand{\variance}{\textsf{Var}}
\newcommand{\covariance}{\textsf{Cov}}
\newcommand{\g}{\mathbf{g}}
\newcommand{\h}{\mathbf{h}}
\newcommand{\bphi}{\boldsymbol{\phi}}
\newcommand{\bpsi}{\boldsymbol{\psi}}
\newcommand{\Tup}{\tau_{\text{up}}}
\newcommand{\Tdp}{\tau_{\text{dp}}}
\newcommand{\Tdd}{\tau_{\text{dd}}}
\newcommand{\Ed}{\mathcal{E}_{\text{d}}}
\newcommand{\Eup}{\mathcal{E}_{\text{up}}}
\newcommand{\Edp}{\mathcal{E}_{\text{dp}}}
\newcommand{\CN}{\mathcal{CN}}
\newcommand{\complex}{\mathbb{C}}
\newcommand{\Tframe}{\tau_{\textrm{frame}}}
\newcommand{\DT}{\text{DT}}
\newcommand{\CSI}{\text{sCSI}}
\newcommand{\SE}{\text{SE}}
\newcommand{\Vmax}{V_{\textrm{max}}}
\newcommand{\I}{\mathbf{I}}
\newcommand{\UE}{\text{UE}}
\newcommand{\z}{\mathbf{z}}

\usepackage{cite}

\ifCLASSINFOpdf
\else
\fi
\usepackage{amsmath}
\usepackage[cmintegrals]{newtxmath}
\begin{document}

\title{Impact of Mobility on Downlink Cell-Free Massive MIMO Systems}

\author{Abhinav~Anand,
        Chandra~R.~Murthy,~\IEEEmembership{Senior Member,~IEEE,}
        and~Ribhu~Chopra,~\IEEEmembership{Member,~IEEE}
\thanks{A. Anand and C. R. Murthy are with the Department of ECE, Indian Institute of Science, Bangalore 560012,
India (e-mails: \{abhinavanand, cmurthy\}@iisc.ac.in). R. Chopra is with the Department of EEE, Indian Institute of Technology, Guwahati 781039, India (email: ribhu@outlook.com).}
}

\maketitle
\vspace{-1cm}
\begin{abstract}

In this paper, we analyze the achievable downlink spectral efficiency of cell-free massive multiple input multiple output (CF-mMIMO) systems, accounting for the effects of channel aging (caused by user mobility) and pilot contamination. We consider two cases, one where user equipments (UEs) rely on downlink pilots beamformed by the access points (APs) to estimate downlink channel, and another where UEs utilize statistical channel state information (CSI) for data decoding. For comparison, we also consider cellular mMIMO and derive its achievable spectral efficiency with channel aging and pilot contamination in the above two cases. Our results show that, in CF-mMIMO, downlink training is preferable over statistical CSI when the length of the data sequence is chosen optimally to maximize the spectral efficiency. In cellular mMIMO, however, either one of the two schemes may be better depending on whether user fairness or sum spectral efficiency is prioritized. Furthermore, the CF-mMIMO system generally outperforms cellular mMIMO even after accounting for the effects of channel aging and pilot contamination. Through numerical results, we illustrate the effect of various system parameters such as the maximum user velocity, uplink/downlink pilot lengths, data duration, network densification, and provide interesting insights into the key differences between cell-free and cellular mMIMO systems.

\end{abstract}


%
\IEEEpeerreviewmaketitle

\section{Introduction}\label{sec:introduction}

\IEEEPARstart{C}{ell}-free massive multiple-input multiple-output (mMIMO) has received considerable attention in recent years \cite{zhang_access_2019, interdonato_eurasip_2019, zhang_jsac_2020}. Originally, the cellular mMIMO architecture, where a large number of antennas colocated on a base station (BS) serve an exclusive set of user equipments (UEs), was shown to bring substantial gains in spectral efficiency (SE) over older generation technologies \cite{marzetta_book, bjrnson_book_2017}. However, in a cellular network, only the UEs that are near the BS, i.e., in the cell center, enjoy high data rates while the UEs at the cell edge experience high inter-cell interference and low throughput. The primary goal of next-generation wireless networks must not be to improve the peak data rate but the rate that can be delivered at a vast majority of UE locations in a given region \cite{demir_textbook_2021}. Cell-free mMIMO has been proposed as a potential solution for providing uniformly high data rates in a wide network.

In a cell-free mMIMO network, a large number of geographically distributed access points (APs) coherently serve multiple UEs on the same time-frequency resource~\cite{ngo_twc_2017}. The APs are connected to one or more central processing units (CPUs) via fronthaul links~\cite{demir_textbook_2021}. The CPU orchestrates the AP operations and jointly processes the signals to/from the UEs. Owing to the distributed implementation, the signal co-processing at multiple APs and the massive number of AP antennas, cell-free mMIMO offers higher coverage probability to UEs than cellular mMIMO~\cite{demir_textbook_2021}. Studies demonstrating the performance advantages of cell-free mMIMO over conventional cellular mMIMO and small-cell networks can be found in~\cite{ngo_twc_2017, nayebi_twc_2017, bjornson_twc_2020, nayebi_asilomar_2015}. The performance gains, however, come at the cost of increased fronthaul requirements which translate to increased energy and power consumption~\cite{ngo_gcom_2018, nguyen_letters_2017, zheng_twc_2021, bashar_tgcn_2019}.

To realize the enormous spatial multiplexing gain offered by mMIMO, each AP needs to estimate the channel from the associated UEs on the uplink. The channel estimates acquired at the AP are used to perform receive combining (or transmit precoding) on the uplink (or downlink) data symbols transmitted subsequently \cite{demir_textbook_2021}. Hence, the quality of the channel estimates strongly affects the performance of the network. Most previous studies on cell-free mMIMO assume that the channel between the AP and the UE is quasi-static or block-fading. Such a model is valid when the UEs in the network are stationary or move slowly. As such, previous results on the achievable SE of cell-free mMIMO may not hold true in extreme mobility scenarios. It is important, therefore, that we investigate how mobility impacts the performance of a cell-free mMIMO network.

User mobility brings two major problems to an mMIMO system. First, the temporal variations in the channel response resulting from user movement cause a disparity between the channel state information (CSI) acquired at the AP and the channel experienced by the data symbols. This is known as \textit{channel aging}~\cite{truong_JCN_2013} and it results in a drop in the achievable network SE. Second, due to the fast-varying nature of the channel, the coherence interval may not be long enough to accommodate pairwise orthogonal pilot sequences for all the UEs present in the network. As a result, a fraction of the UEs end up sharing the same pilot sequence and contaminate each other's channel estimate. This phenomenon, known as \textit{pilot contamination}~\cite{bjrnson_book_2017, demir_textbook_2021} degrades the quality of the available channel estimates and causes a drop in the achievable SE.

In cellular mMIMO, when a signal is transmitted from a large number of BS antennas, the effective downlink channel after transmit precoding tends to converge to its mean value. This is known as \textit{channel hardening} \cite{marzetta_book}. An important consequence of this phenomenon is that UEs in cellular mMIMO do not need to estimate the downlink channel, and they can rely instead on knowledge of the channel statistics to decode the data symbols. This eliminates the need for downlink pilots and makes cellular mMIMO scalable with respect to the number of BS antennas. In cell-free mMIMO, however, the transmitting AP antennas are distributed over a wide region and a UE experiences large path loss from the far-away APs. As a result, the channel hardening phenomenon in cell-free mMIMO is much less pronounced than in cellular mMIMO. In this regard, the authors in \cite{interdonato_twc_2019} showed that downlink training using beamformed pilots can significantly improve the  downlink SE of cell-free mMIMO, outweighing the additional training overhead cost. Moreover, it was shown in \cite{polegre_icc_2020, chen_tcom_2018} that one should not rely on channel hardening when analyzing the performance of or designing receiver algorithms for cell-free mMIMO networks. These studies, however, do not account for the time-varying nature of the channel arising from user mobility. In the context of single-cell orthogonal frequency division multiplexing (OFDM)-based mMIMO communications, in \cite{Abhinav_TCOM_2022}, we analyzed the effect of pilot contamination and channel aging on the SE, and developed pilot and data subcarrier allocation schemes to improve the SE.

The authors in \cite{zheng_twc_2021} studied the uplink and the downlink achievable SE of cell-free mMIMO in the presence of channel aging and pilot contamination assuming large-scale fading decoding (LSFD) and matched filtering (MF) receivers. The performance analysis of zero-forcing (ZF) precoding in mobility-impaired downlink cell-free mMIMO was taken up in \cite{jiang_letters_2021}. In \cite{chopra_letters_2021}, a model involving varying rates of channel evolution across APs was proposed and the performance of uplink cell-free mMIMO was analyzed. However, none of the above works consider the use of downlink training in cell-free mMIMO. Thus, the question of whether the use of beamformed pilots can improve the downlink SE in the face of user mobility and time-varying channels remains open in the literature, and is the focus of this work. Furthermore, even in the context of cellular mMIMO, it is unclear whether the channel hardening effect is sufficient to extract the benefits of mMIMO under user mobility. We study this aspect also, in this paper.

\subsection{Our Contributions}

In this paper, we derive analytical expressions for the approximate achievable downlink SE of both cell-free and cellular mMIMO systems impaired by channel aging and pilot contamination. We incorporate time-varying channel and non-orthogonal uplink/downlink pilots in the analysis. We investigate the effect of factors such as the maximum user velocity, the relative uplink/downlink training lengths and the downlink data duration on the downlink SE of cell-free mMIMO. We find that while the downlink training scheme outperforms the statistical CSI scheme, the relative gain in performance reduces at higher user mobility. Moreover, as user mobility increases, the channel varies more rapidly, and it is necessary to shorten the data duration and re-estimate the channel more frequently. Nonetheless, despite the additional training overhead, downlink training outperforms statistical CSI when the data duration is chosen optimally. We also look at the effect of densifying a cell-free network with APs when the total number of AP antennas in the network is held fixed. We find that the performance gain due to densification in the presence of downlink training is much more significant compared to when UEs rely on statistical CSI to decode downlink data. From a sum-SE perspective, however, this gain in performance diminishes as the UEs move faster. Finally, focusing on cellular mMIMO, we show that the gain in the average 90\% likely downlink SE due to downlink training is marginally negative owing to the higher degree of channel hardening. However, there is still substantial gain in the sum-SE with downlink training. As such, either of the two schemes may be considered in cellular mMIMO depending on which performance measure is important.

The rest of the paper is organized as follows. Section \ref{sec:system_model} discusses the system model and describes the frame structure. Section \ref{sec:performance_cell_free} presents the analytical results on the approximate downlink SE and with downlink training and statistical CSI-based decoding. Section \ref{sec:performance_cellular} presents the downlink SE analysis for a cellular mMIMO network. Section \ref{sec:numerical_results} presents numerical results to elucidate the performances of cell-free and cellular mMIMO under user mobility. Finally, section \ref{sec:conclusion} concludes the paper and provides suggestions for future work.

\textit{Notation:} Matrices and column vectors are denoted by boldface uppercase and lowercase letters. The notations $(\cdot)^{*}$, $(\cdot)^{\transpose}$ and $(\cdot)^{\hermitian}$ represent the conjugate, the transpose and the conjugate transpose operation. The symbols $\mathbf{0}_{L}$ and $\mathbf{I}_{L}$ denote the null matrix and the identity matrix of order $L$. The notation $\CN\left(\mathbf{0}_{L}, \mathbf{R}\right)$ refers to an $L$-dimensional circularly symmetric complex normal distribution with mean vector $\mathbf{0}_{L}$ and covariance matrix $\mathbf{R}$. The notations $\expect\{\cdot\}$, $\covariance\{\cdot\}$ and $\variance\{\cdot\}$ represent the expectation, the covariance and the variance operations.

\section{System Model}\label{sec:system_model}
We consider a time-division duplex (TDD) cell-free mMIMO network in which $M$ APs (indexed as $m = 1,2, \dots, M$) each equipped with $L$ antennas (indexed as $l = 1,2, \dots, L$) coherently serve $K$ single-antenna mobile UEs (indexed as $k = 1, 2, \dots, K)$ on the same time-frequency resource. As is common in the mMIMO literature, we assume $ML \gg K$. The APs are geographically distributed over a wide region and are connected to a CPU via ideal fronthaul links. We denote the velocity of the $k$th UE by $v_{k}$, and assume $0 \leq v_k \leq V_{\text{max}}$ where $V_{\text{max}}$ denotes the maximum possible UE velocity in the system~\cite{zhang_tcom_2017}. Furthermore, the UEs move independently of each other. Figure \ref{fig:network_figure} gives an illustration of the network.

\begin{figure}
    \centering
    \includegraphics[width=3.0in]{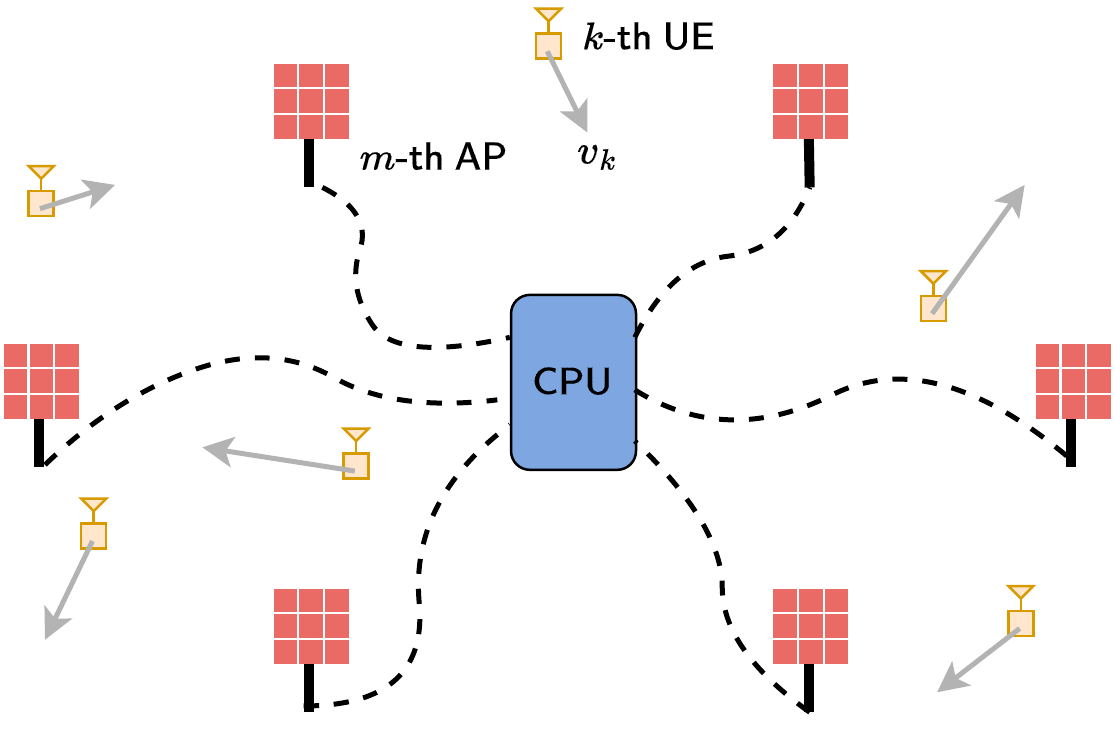}
    \caption{A cell-free mMIMO network with mobile UEs.}
    \label{fig:network_figure}
\end{figure}

User mobility causes the channel coefficients between the UEs and the APs to vary with time. For the purpose of this work, we consider a transmit frame comprising $\Tframe$ contiguous symbols (indexed as $n = 0, 1, \dots, \Tframe - 1$) each of which may be used for either pilot or data transmission. The channel coefficients between an AP and a UE
are assumed to remain constant within one symbol; they may however vary from symbol to symbol. 
The time duration of each symbol is $T_{\textrm{s}} = \frac{1}{B}$, where $B$ denotes the system bandwidth.

Let $\g_{mk}[n] = \sqrt{\beta_{mk}}\h_{mk}[n] \in \mathbb{C}^L$ denote the complex-valued channel vector between the $m$th AP and the $k$th UE during the $n$th symbol. Here, $\h_{mk}[n] \sim \CN\left(\mathbf{0}_{L}, \mathbf{I}_{L}\right)$ denotes the small-scale independent and identically distributed (i.i.d.) Rayleigh fading component between the $m$th AP and the $k$th UE. The quantity $\beta_{mk}$ is the large-scale fading coefficient that models the path-loss and the shadowing effects. We assume that $\beta_{mk}$ is constant across all antennas of the $m$th AP, across all symbols in a transmit frame, and is known at the APs and the CPU.

The temporal variation of the propagation channel between the $m$th AP and the $k$th UE is modeled as follows: starting from the channel at the zeroth instant $\g_{mk}[0]$, the channel at a later instant $n$ ($n \geq 0$) is expressed as \cite{chopra_twc_2018, zheng_twc_2021}:
\begin{align}\label{eq:ChannelAging}
    \g_{mk}[n] = \rho_{k}[n] \g_{mk}[0] + \bar{\rho}_{k}[n] \mathbf{z}_{mk}[n],
\end{align}
where $\rho_{k}[n] = J_{0}\left(\frac{2\pi v_{k} f_{\text{c}} n T_{\text{s}}}{c}\right)$ denotes the Jakes' autocorrelation between $\g_{mk}[0]$ and $\g_{mk}[n]$ with $J_{0}(.)$ representing the zeroth order Bessel function of the first kind, $f_{\text{c}}$ the carrier frequency, $c$ the speed of light, and $T_{\text{s}}$ the symbol duration, and $\bar{\rho}_{k}[n] = \sqrt{1 - \rho_{k}^{2}[n]}$. The quantity $\mathbf{z}_{mk}[n]$ in \eqref{eq:ChannelAging} represents the innovation component due to channel aging which is independent of and identically distributed as $\g_{mk}[0]$, i.e., as $\CN\left(\mathbf{0}_{L}, \beta_{mk}\mathbf{I}_{L}\right)$. 

The transmit frame comprises three successive phases: uplink training, downlink training and downlink data transmission (see Figure \ref{fig:frame_structure}.) Since the overall training duration is typically small, we can assume that the channel coefficients remain approximately constant during the training intervals\footnote{For example, in a system with bandwidth $B=1$ MHz at a center frequency of $2$ GHz and with the users moving at a maximum velocity of $100$ km/hr, if the total training duration spans $20$ symbols, the Jakes' autocorrelation between the start and end of the training duration is about $0.99986$.}~\cite{papa_tvt_2017, papa_twc_2015, kong_tcom_2015}. 

\begin{figure}
    \centering
    \includegraphics{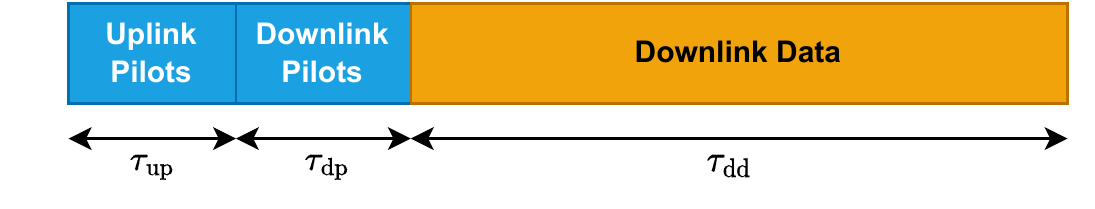}
    \caption{The transmit frame structure considered in this work.}
    \label{fig:frame_structure}
\end{figure}

\subsection{Uplink Training}
In the first phase, the UEs transmit pilot sequences using which the APs estimate the uplink channel $\g_{mk} = \g_{mk}[0]$. We let $\Tup < K$ denote the number of mutually orthogonal $\Tup$-length pilot sequences available for transmission. The $k$th UE transmits the pilot sequence $\sqrt{\Tup}\bphi_{k}^{\hermitian}$ where $||\bphi_{k}||^{2} = 1$. Since $\Tup < K$, two or more UEs may transmit same pilot sequence. Therefore, for two UEs $k$ and $k'$, $\bphi_{k}^{\hermitian} \bphi_{k'}$ equals $1$ when UEs $k$ and $k'$ transmit the same pilot sequence, and equals $0$ otherwise. The training signal received at the $m$th AP is an $L\times\Tup$ matrix expressed as
\begin{equation}\label{eq:ULReceivePilot}
    \mathbf{Y}_{\text{up}, m} = \sqrt{\Tup\Eup} \sum_{k'=1}^{K} \g_{mk'} \bphi_{k'}^{\hermitian} + \mathbf{W}_{\text{up}, m} \in \complex^{L\times\Tup}.
\end{equation}
In the above equation, $\Eup$ denotes the normalized transmit SNR for the uplink and $\mathbf{W}_{\text{up}, m} \in \complex^{L\times\Tup}$ denotes noise at the $m$th AP whose elements are i.i.d $\CN\left(0, 1\right)$. To estimate the uplink channel from UE $k$, the $m$th AP correlates the received pilot signal with pilot $\bphi_{k}$ as
\begin{align}
    \breve{\mathbf{y}}_{\text{up}, mk} =& \mathbf{Y}_{\text{up}, m}  \bphi_{k}  = \sqrt{\Tup\Eup} \g_{mk} + \sqrt{\Tup\Eup} \sum_{k' \neq k}^{K} \g_{mk'} \bphi_{k'}^{\hermitian} \bphi_{k} + \mathbf{w}_{\text{up}, mk},  \label{eq:ULPilotProcessed}
\end{align}
where $\mathbf{w}_{\text{up}, mk} = \mathbf{W}_{\text{up}, m} \bphi_{k}$ has i.i.d. $\CN\left(0, 1\right)$ entries. Now, using  $\breve{\mathbf{y}}_{\text{up}, mk}$, the $m$th AP obtains an MMSE estimate $\hat{\g}_{mk}$ of $\g_{mk}$, as $\hat{\g}_{mk} = c_{mk} \breve{\mathbf{y}}_{\text{up}, mk}$, where
\begin{equation}\label{eq:gammaConst}
    c_{mk} \triangleq \frac{\sqrt{\Tup\Eup} \beta_{mk}}{\Tup\Eup\sum_{k' = 1}^{K} \beta_{mk'} |\bphi_{k'}^{\hermitian}\bphi_{k}|^{2} + 1}.
\end{equation}
It is known that, with MMSE estimation,  $\hat{\g}_{mk} \sim \CN(0, \gamma_{mk}\mathbf{I}_{L})$, with $\gamma_{mk} = \sqrt{\Tup\Eup}c_{mk}\beta_{mk}$.
Also, the channel estimation error  $\Tilde{\g}_{mk} = \g_{mk} - \hat{\g}_{mk}$ is independent of  $\hat{\g}_{mk}$ and is distributed as $\CN\left(0, \left(\beta_{mk} - \gamma_{mk}\right)\mathbf{I}_{L}\right)$.

\subsection{Downlink Training}
Having obtained the channel estimate on the uplink, the APs precode and transmit pilot sequences to UEs in the downlink direction. We consider maximum-ratio (MR) precoding at the APs, also known as conjugate beamforming or matched filtering (MF) precoding, as it allows the APs to perform channel estimation and precoding locally without sharing their CSI with the CPU \cite{ngo_twc_2017, interdonato_twc_2019}. We note that the subsequent analysis can be extended to other linear precoding schemes as well, with some more bookkeeping. Similar to the uplink case, we assume that there are $\Tdp < K$ mutually orthogonal downlink pilot sequences in total. Let $\sqrt{\Tdp}\bpsi_{k}^{\hermitian} \in \complex^{1 \times\Tdp}$ denote the pilot sequence transmitted for the $k$th UE, with $\|\bpsi_{k}\|^{2} = 1$.

With MR precoding, the $m$th AP transmits $\Tdp$ training symbols over its $L$ antennas represented by the matrix $\mathbf{X}_{\text{dp}, m} = \sqrt{\Tdp\Edp} \sum_{k=1}^{K} \sqrt{\eta_{mk}} \hat{\g}_{mk}^{*}\bpsi_{k}^{\hermitian} \in \mathbb{C}^{L \times \Tdp}$. Here, $\eta_{mk} \in [0,1]$ denotes the power control coefficient used by the $m$th AP for its transmissions to UE $k$, and $\Edp$ denotes the normalized downlink transmit SNR. The total power spent by the $l$th antenna of the $m$th AP on downlink pilots is 
\begin{align}\label{eq:DLTotalPilotPower}
    \expect\Big\{\big|\big|\mathbf{X}_{\text{dp}, m_{l}}\big|\big|^{2}\Big\} =&  \Tdp\Edp \sum_{k=1}^{K} \eta_{mk} \gamma_{mk}  + \Tdp\Edp \sum_{k=1}^{K}\sum_{k'\neq k}^{K} \sqrt{\eta_{mk}\eta_{mk'}}\bpsi_{k'}^{\hermitian}\bpsi_{k} \expect\{\hat{g}_{m_{l}k} \hat{g}_{m_{l}k'}^{*}\}.
\end{align}
The $\expect\{\hat{g}_{m_{l}k} \hat{g}_{m_{l}k'}^{*}\}$ term in \eqref{eq:DLTotalPilotPower}  will be non-zero if UEs $k$ and $k'$ transmit the same uplink pilot sequence. As shown in \cite{interdonato_twc_2019}, it can be eliminated by assigning orthogonal pilots to UEs that share the same pilot on the uplink, i.e., if $\bphi_{k'} = \bphi_{k}$, then  $\bpsi_{k'}^{\hermitian} \bpsi_{k} = 0$. Such an assignment is feasible as long as $\Tup\Tdp \geq K$. Then, the total power spent by the $m$th AP on downlink pilots equals $\Tdp\Edp\sum_{k=1}^{K}L\eta_{mk}\gamma_{mk}$. Since the total available downlink training power is $\Tdp\Edp$, we get the following constraint on $\eta_{mk}$:
\begin{align}\label{eq:DLPilotPowerConstraint}
    \sum_{k=1}^{K} \eta_{mk}\gamma_{mk} \leq \frac{1}{L}.
\end{align}

Now, the downlink pilot signal received by the $k$th UE is
\begin{align}
    \mathbf{y}_{\text{dp}, k} = \sum_{m=1}^{M} \g_{mk}^{\transpose} \mathbf{X}_{\text{dp}, m} + \mathbf{w}_{\text{dp}, k}
    = \sqrt{\Tdp\Edp} \sum_{m=1}^{M}\sum_{k'=1}^{K} \sqrt{\eta_{mk}} \g_{mk}^{\transpose} \hat{\g}_{mk'}^{*} \bpsi_{k'}^{\hermitian} + \mathbf{w}_{\text{dp}, k}.\label{eq:DLPilotReceive}
\end{align}
where $\mathbf{w}_{\text{dp}, k}$ containing i.i.d $\CN\left(0,1\right)$ entries represents noise at the $k$th UE. To estimate the downlink channel, the $k$th UE correlates the recieved signal with $\bpsi_{k}$ to obtain
\begin{equation}\label{eq:DLPilotFinal}
    \breve{y}_{\text{dp}, k} = \sqrt{\Tdp\Edp}d_{kk} + \sqrt{\Tdp\Edp}\sum_{k' \neq k}^{K} d_{kk'}\bpsi_{k'}^{\hermitian}\bpsi_{k} + \Tilde{w}_{\text{dp}, k}.
\end{equation}
where $\Tilde{w}_{\text{dp}, k} = \mathbf{w}_{\text{dp}, k}\bpsi_{k}$, and
$d_{kk'} \triangleq \sum_{m=1}^{M} \sqrt{\eta_{mk'}} \g_{mk}^{\transpose}\hat{\g}_{mk'}^{*}$ represents the effective downlink channel experienced by the $k$th user for the data stream intended to the $k'$th user. Note that $d_{kk}$ is the desired downlink channel coefficient. The term  containing $\bpsi_{k'}^{\hermitian}\bpsi_{k}$ represents the interference due to downlink pilot contamination. Using $\breve{y}_{\text{dp}, k}$, the $k$th UE computes an MMSE estimate $\hat{d}_{kk}$ of $d_{kk}$ as~\cite{interdonato_tcom_2021}
\begin{align}\label{eq:DLChannelEstimate}
    \hat{d}_{kk} = \sum_{m=1}^{M}L\sqrt{\eta_{mk}} \gamma_{mk}  + &  \frac{\sqrt{\Tdp\Edp} \sum_{m=1}^{M} L\eta_{mk}\gamma_{mk}\beta_{mk}}{1 + \Tdp\Edp \sum_{m=1}^{M} \sum_{k'=1}^{K} L \eta_{mk'} \gamma_{mk'} \beta_{mk} |\bpsi_{k'}^{\hermitian}\bpsi_{k}|^{2}} \nonumber \\
    \times& \left(\breve{y}_{\text{dp}, k} - \sqrt{\Tdp\Edp} \sum_{m=1}^{M} L\sqrt{\eta_{mk}} \gamma_{mk}\right).
\end{align}
Due to MMSE estimation, we can write  $d_{kk} = \hat{d}_{kk} + \Tilde{d}_{kk}$ where $\Tilde{d}_{kk}$ is the zero-mean downlink channel estimation error; note that $\hat{d}_{kk}$ and $\Tilde{d}_{kk}$ are  uncorrelated.

\subsection{Downlink Data Transmission}
After gaining knowledge of the effective downlink channel, the UEs proceed to detect the incoming data symbols. The data transmission phase is assumed to be $\Tdd$ symbols long; thus $\Tframe = \Tup + \Tdp + \Tdd$. During the $n$th signaling interval ($n = \Tup + \Tdp, \dots, \Tframe-1$), the $m$th AP applies MR precoding on the data symbols and transmits an $L$-dimensional vector given by 
\begin{align}
    \boldsymbol{x}_{m}[n] = \sqrt{\Ed} \sum_{k=1}^{K} \sqrt{\eta_{mk}}\hat{\g}_{mk}^{*} q_{k}[n],
\end{align}
where $\Ed$ denotes the normalized transmit SNR for downlink data and $q_{k}[n]$ denotes the $n$th data symbol transmitted for the $k$th UE. The symbols $\{q_{k}[n]\}$ are assumed to be uncorrelated across all the UEs. Further, they are assumed to have zero mean and unit variance, i.e., $\expect\{|q_{k}[n]|^{2}\} = 1$. Similar to the downlink training case, the total power spent by the $m$th AP during the $n$th transmission is $\expect\big\{\big|\big|\boldsymbol{x}_{m}[n]\big|\big|^{2}\big\} = \Ed\sum_{m=1}^{M}L\eta_{mk}\gamma_{mk}$ whose maximum value is $\Ed$. Therefore, the data power constraint at the $m$th AP is $\sum_{k=1}^{K}\eta_{mk}\gamma_{mk} \leq \frac{1}{L}$, the same as the pilot power constraint given in \eqref{eq:DLPilotPowerConstraint}. Now, the $n$th interval data signal received at the $k$th UE is given by
\begin{align}
    r_{\text{d}, k}[n] =& \sum_{m=1}^{M} \g_{mk}^{\transpose}[n] \boldsymbol{x}_{m}[n] + w_{\text{d},k}[n] \nonumber \\
    =& \sqrt{\Ed} \sum_{m=1}^{M} \sqrt{\eta_{mk}} \g_{mk}^{\transpose}[n] \hat{\g}_{mk}^{*}q_{k}[n]  + \sqrt{\Ed} \sum_{m=1}^{M} \sum_{k' \neq k}^{K} \sqrt{\eta_{mk'}}\g_{mk}^{\transpose}[n]\hat{\g}_{mk'}^{*}q_{k'}[n] + w_{\text{d},k}[n], \label{eq:DLDataReceive1} 
\end{align}
where $w_{\text{d},k}[n] \sim \CN\left(0,1\right)$ denotes noise at the $k$th UE. The above can be re-written as
\begin{align}\label{eq:DLDataReceive2}
    r_{\text{d}, k}[n] =  \sqrt{\Ed}d_{kk}[n] q_{k}[n] + \sqrt{\Ed} \sum_{k' \neq k}^{K} d_{kk'}[n] q_{k'}[n] + w_{\text{d},k}[n]
\end{align}
where 
\begin{align}\label{eq:NthDLChannelGain}
    d_{kk'}[n] \triangleq \sum_{m=1}^{M} \sqrt{\eta_{mk'}} \g_{mk}^{\transpose}[n]\hat{\g}_{mk'}^{*}
\end{align}
denotes the effective downlink channel coefficient at the $n$th-instant. We note that substituting $n = 0$ in \eqref{eq:NthDLChannelGain} gives $d_{kk'}[0] = \sum_{m=1}^{M} \sqrt{\eta_{mk'}} \g_{mk}^{\transpose}[0]\hat{\g}_{mk'}^{*}$. Since $\g_{mk}[0] = \g_{mk}$, we have $d_{kk'}[0] = d_{kk'}$ which is the downlink channel coefficient during downlink training. Using the channel-aging model in \eqref{eq:ChannelAging}, the quantity $d_{kk'}[n]$ can be expressed as
\begin{align}
    d_{kk'}[n] =& \rho_{k}[n]\sum_{m=1}^{M} \sqrt{\eta_{mk'}} \g_{mk}^{\transpose} \hat{\g}_{mk'}^{*} + \bar{\rho}_{k}[n]\sum_{m=1}^{M} \sqrt{\eta_{mk'}} \mathbf{z}_{mk}^{\transpose}[n]\hat{\g}_{mk'}^{*} \nonumber \\
    =& \rho_{k}[n]d_{kk'} + \bar{\rho}_{k}[n]z_{kk'}[n] \nonumber \\
    =& \rho_{k}[n]\hat{d}_{kk'} + \rho_{k}[n]\Tilde{d}_{kk'} + \bar{\rho}_{k}[n]z_{kk'}[n], \label{eq:d_kk'[n]}
\end{align}
where $z_{kk'}[n] = \sum_{m=1}^{M}\sqrt{\eta_{mk'}}\mathbf{z}_{mk}^{\transpose}[n]\hat{\g}_{mk'}^{*}$ represents the $n$th-instant innovation component in the downlink channel due to channel aging. We note that the quantities $\hat{d}_{kk'}$, $\Tilde{d}_{kk'}$ and $z_{kk'}[n]$ in \eqref{eq:d_kk'[n]} are mutually uncorrelated. However, the presence of the innovation component in the downlink channel entails significant bookkeeping in deriving the downlink SE expressions, which is the focus of this work.

The downlink channel coefficient defined in \eqref{eq:NthDLChannelGain} is a non-Gaussian quantity. However, when the total number of AP antennas $ML$ is sufficiently large, it approximates a complex Gaussian random variable with mean, variance and pseudo-variance \cite{park_book} given by
\begin{align}
    \mu_{kk'}[n] =& \rho_{k}[n]\sum_{m=1}^{M}L\sqrt{\eta_{mk'}}\gamma_{mk'}\frac{\beta_{mk}}{\beta_{mk'}}\bphi_{k'}^{\hermitian}\bphi_{k}\\
    \varsigma_{kk'} =& \sum_{m=1}^{M}L\eta_{mk'}\gamma_{mk'}\beta_{mk}\\
    \varrho_{kk'}[n] =& \rho_{k}^{2}[n]\sum_{m=1}^{M}L\eta_{mk'}\gamma_{mk'}^{2}\left(\frac{\beta_{mk}}{\beta_{mk'}}\right)^{2}\left(\bphi_{k'}^{\hermitian}\bphi_{k}\right)^{2},
\end{align}
where the notations $\mu_{kk'}[n]$, $\varsigma_{kk'}$ and $\varrho_{kk'}[n]$ denote the mean, variance and the pseudo-variance of $d_{kk'}[n]$, respectively. Figure \ref{fig:KLdivergence} shows the Kullback-Leibler (KL) distance \cite{cover_book} between the simulated and the Gaussian probability density functions (PDFs) of the real and the imaginary parts of $d_{kk'}[0]$ and $d_{kk}[0]$ when $M = 100$ APs and $K = 40$ UEs are deployed in a cell-free network. The propagation model is adopted from \cite{bjornson_twc_2020} and the simulation settings are provided in section \ref{sec:numerical_results}. We see that the KL distance between the two PDFs reduces as the number of antennas on an AP grows. Therefore, for finite values of $M$ and $L$, the downlink channel gain can be treated as approximately complex Gaussian.\footnote{Note, however, that unlike the uplink channel coefficient, the downlink channel gain is not circularly symmetric owing to nonzero mean and pseudo-variance.} As a consequence, $\hat{d}_{kk'}$, $\Tilde{d}_{kk'}$ and $z_{kk'}[n]$ in \eqref{eq:d_kk'[n]} become jointly Gaussian.

\begin{figure}
    \centering
    \includegraphics[width=4.5in]{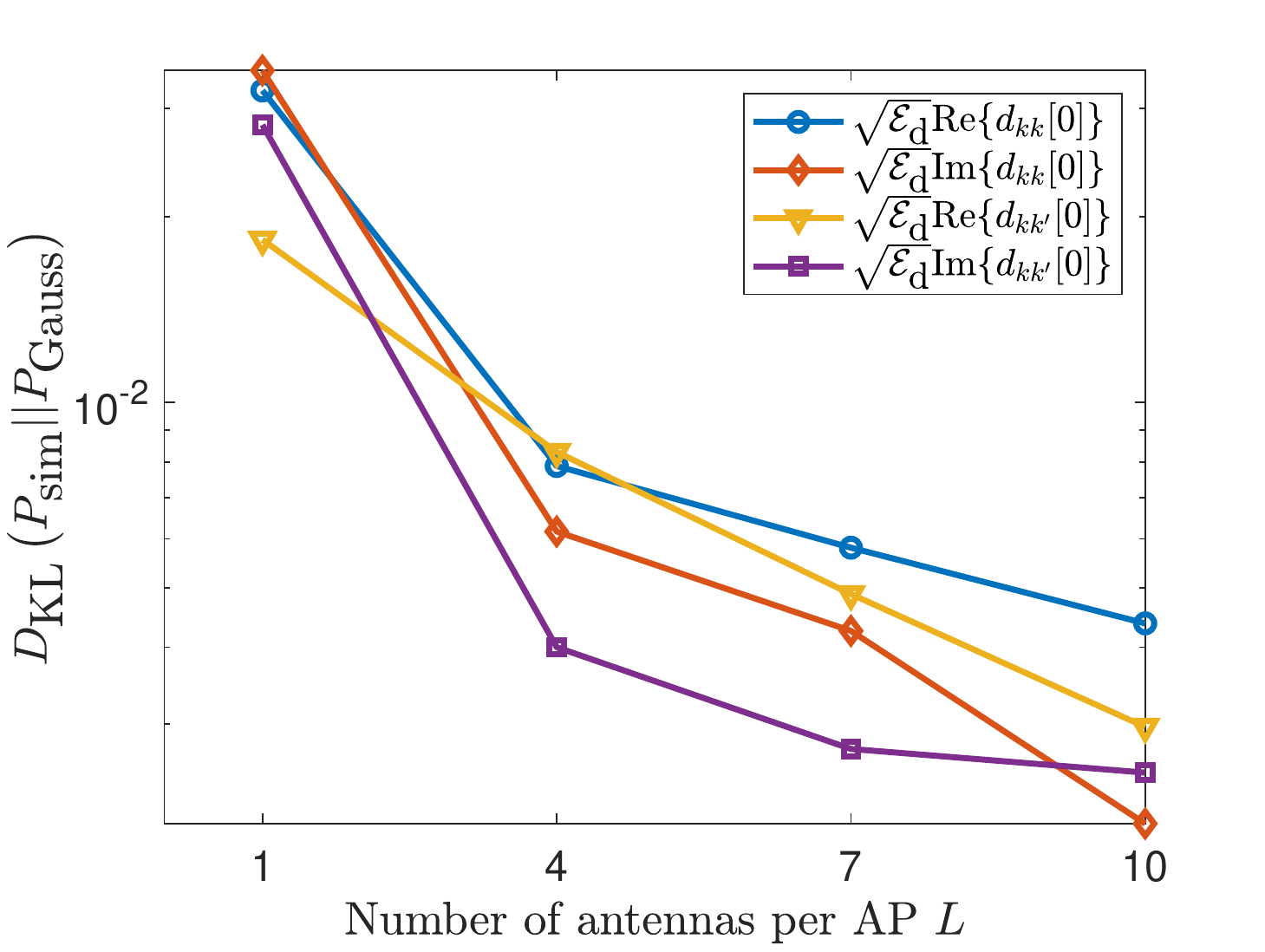}
    \caption{Plot of the KL distance between the simulated ($P_{\textrm{sim}}$) and the Gaussian ($P_{\textrm{Gauss}}$) distributions of the downlink channel coefficient as a function of the number of antennas on an AP.}
    \label{fig:KLdivergence}
\end{figure}

\section{Performance of Cell-Free mMIMO}\label{sec:performance_cell_free}

In this section, we derive a closed-form expression for the approximate achievable downlink SE of the cell-free massive MIMO network considering channel aging and uplink/downlink pilot contamination effects. We compare it against the scenario when UEs rely on statistical CSI to recover the transmitted data symbols. 

\subsection{Performance With Downlink Training}
\begin{proposition}
The $n$th-instant approximate achievable downlink SE of the $k$th UE in the cell-free mMIMO network described above takes the form
\begin{align}\label{eq:CF_SE_DT}
    \SE_{k}^{\text{CF},\DT}[n] = \log_{2}\left(1 + \text{SINR}_{k}^{\text{CF},\DT}[n]\right)
\end{align}
where $\text{SINR}_{k}^{\text{CF},\DT}[n]$ denotes the $n$th-instant effective downlink SINR of the $k$th UE with downlink training, given by
\begin{multline}\label{eq:CF_SINR_DT}
    \text{SINR}_{k}^{\text{CF},\DT}[n] \\= \frac{\rho_{k}^{2}[n]\Ed\left(\sum_{m=1}^{M}L\sqrt{\eta_{mk}}\gamma_{mk}\right)^{2} + \rho_{k}^{2}[n]\Ed\kappa_{k}}{\Ed\sum_{k'\neq k}\left(\varsigma_{kk'} + \rho_{k}^{2}[n]\left(\sum_{m=1}^{M}L\sqrt{\eta_{mk'}}\gamma_{mk'}\frac{\beta_{mk}}{\beta_{mk'}}\right)^{2}|\bphi_{k'}^{\hermitian} \bphi_{k}|^{2}\right) + \Ed\left(\varsigma_{kk} - \rho_{k}^{2}[n]\kappa_{k}\right) + 1},
\end{multline}
in which $\varsigma_{kk'} = \sum_{m=1}^{M}L\eta_{mk'}\gamma_{mk'}\beta_{mk}$ denotes the variance of the effective downlink channel $d_{kk'}[n]~\forall~n$ and $\kappa_{k} = \frac{\Tdp\Edp\left(\sum_{m=1}^{M}L\eta_{mk}\gamma_{mk}\beta_{mk}\right)^{2}}{1 +  \Tdp\Edp \sum_{m=1}^{M} \sum_{k'=1}^{K}L \eta_{mk'} \gamma_{mk'} \beta_{mk} |\bpsi_{k'}^{\hermitian}\bpsi_{k}|^{2}}$ denotes the variance of the downlink channel estimate $\hat{d}_{kk}$.
\end{proposition}

\begin{proof}
See Appendix \ref{app:proof_proposition_1}.
\end{proof}

From the numerator in the SINR expression, it is clear that the coherent beamforming gain decreases with increasing transmission index $n$ which is the effect of channel aging. The first term in the denominator of the SINR represents the multi-user interference due to channel aging and uplink pilot contamination. The term $\Ed\left(\varsigma_{kk} - \rho_{k}^{2}[n]\kappa_{k}\right)$ represents the variance of the error in the downlink CSI available at the receiver due to channel estimation and aging, and it increases with $n$. The last term represents the variance of the normalized noise. Thus, we see that user mobility not only degrades the coherent beamforming gain but also incurs additional multi-user interference caused by channel aging and pilot contamination. The result is that the downlink SE decreases at higher transmission indices.

Note that by plugging $n = 0$ in \eqref{eq:CF_SE_DT}, we obtain the expression for the approximate downlink SE when the effect of a time-varying channel is ignored, which matches with the result in~\cite{interdonato_tcom_2021}.

Since the downlink SE in \eqref{eq:CF_SE_DT} varies with the index $n$, it is useful to define an average measure of the downlink SE across a set of symbol transmissions. For a cell-free mMIMO network supported by downlink training, the \textit{average downlink SE} across $\Tframe$ symbols is defined as
\begin{align}\label{eq:AV_SE_DT}
    \overline{\SE}_{k}^{\text{CF},\DT}[\Tframe] \triangleq \frac{\Tdd}{\Tframe}\left(\frac{1}{\Tdd}\sum_{n=\Tup + \Tdp}^{\Tframe-1} \SE_{k}^{\text{CF},\DT}[n]\right).
\end{align}
Based on the above definition, the \textit{average sum-SE} of the cell-free network can be computed as
\begin{align}\label{eq:AV_SUM_SE_DT}
    \overline{\SE}_{\text{sum}}^{\text{CF},\DT}[\Tframe] \triangleq \sum_{k=1}^{K}\overline{\SE}_{k}^{\text{CF},\DT}[\Tframe].
\end{align}

\subsection{Performance With Statistical CSI}
Under the assumptions of a time-varying channel, non-orthogonal uplink pilots, maximum-ratio precoding, and i.i.d Rayleigh fading, a closed-form expression for the $n$th-instant approximate achievable downlink SE of the $k$th UE relying on statistical CSI in a cell-free massive MIMO network was derived in \cite{zheng_twc_2021}. We re-write it below using our notations,
\begin{align}\label{eq:CF_SE_CSI}
    \SE_{k}^{\CSI}[n] = \log_{2}\left(1 + \text{SINR}_{k}^{\text{CF},\CSI}[n]\right),
\end{align}
where $\text{SINR}_{k}^{\text{CF},\CSI}[n]$ denotes the $n$th-instant effective downlink SINR of the $k$th UE relying on statistical CSI and is given by
\begin{align}\label{eq:CF_SINR_CSI}
    \text{SINR}_{k}^{\text{CF},\CSI}[n] = \frac{\rho_{k}^{2}[n]\Ed\left(\sum_{m=1}^{M}L\sqrt{\eta_{mk}}\gamma_{mk}\right)^{2}}{\Ed\sum_{k'\neq k}\left(\varsigma_{kk'} + \rho_{k}^{2}[n]\left(\sum_{m=1}^{M}L\sqrt{\eta_{mk'}}\gamma_{mk'}\frac{\beta_{mk}}{\beta_{mk'}}\right)^{2}|\bphi_{k'}^{\hermitian} \bphi_{k}|^{2}\right) + \Ed\varsigma_{kk} + 1},
\end{align}
in which $\varsigma_{kk'} = \sum_{m=1}^{M}L\eta_{mk'}\gamma_{mk'}\beta_{mk}$ denotes the variance of the effective downlink channel $d_{kk'}[n]~\forall~n$.

Comparing the SINRs in \eqref{eq:CF_SINR_CSI} and \eqref{eq:CF_SINR_DT}, we observe that the two expressions differ by the term $\rho_{k}^{2}[n]\Ed\kappa_{k}$ which gets added to the numerator and subtracted from the denominator in \eqref{eq:CF_SINR_DT}. The quantity $\kappa_{k}$ represents the gain introduced due to downlink training and it embeds the effect of the downlink pilot contamination \cite{interdonato_twc_2019}. The multiplication with $\rho_{k}^{2}[n]$ signifies that owing to channel aging, the additional gain due to downlink training diminishes at higher transmission indices and/or higher user mobility.

We define the average downlink SE of the $k$th UE in a cell-free mMIMO network relying on statistical CSI as
\begin{align}\label{eq:AV_SE_CSI}
    \overline{\SE}_{k}^{\text{CF},\CSI}[\Tframe] \triangleq \frac{1}{\Tframe}\sum_{n=\Tup}^{\Tframe-1} \SE_{k}^{\text{CF},\CSI}[n],
\end{align}
where $\Tframe = \Tup + \Tdd$. The average sum-SE of the network is computed as
\begin{align}\label{eq:AV_SUM_SE_CSI}
    \overline{\SE}_{\text{sum}}^{\text{CF},\CSI}[\Tframe] \triangleq \sum_{k=1}^{K}\overline{\SE}_{k}^{\text{CF},\CSI}[\Tframe].
\end{align}

\section{Performance of Cellular mMIMO}
\label{sec:performance_cellular}

The downlink SE of a mobile UE in a cellular mMIMO network can be analyzed using a similar approach as the above. Such an analysis is not available in the literature, and we present it in this section. Consider a multi-cell mMIMO network comprising $L_{c}$ cells. Each BS is equipped with $M_{c}$ antennas and serves $K_{c}$ UEs. In total, there are $K$ UEs moving in the network. Thus, $K = L_{c}K_{c}$. We denote the $k$th UE ($k = 1, \dots, K_{c}$) in cell $l$ ($l = 1, \dots, L_{c}$) as $\textrm{UE}_{lk}$. We assume that the UEs in each cell are assigned mutually orthogonal pilots on the uplink and the downlink and that the $k$th UE in each cell uses the same pilot sequence (i.e., pilot reuse one). Thus, there will exist only inter-cell pilot contamination. The channel between BS $j$ and $\UE_{lk}$ on the $n$th transmission is modeled as $\g_{lk}^{j}[n] = \rho_{lk}[n]\g_{lk}^{j}[0] + \bar{\rho}_{lk}[n]\z_{lk}^{j}[n]$ where $\rho_{lk}[n] = J_{0}\left(2 \pi f_{\textrm{c}} v_{lk} T_{\textrm{s}}/c\right)$ represents the temporal correlation coefficient of $\UE_{lk}$ at the $n$th instant with the relative velocity of $\UE_{lk}$ denoted by  $v_{lk}$, $\bar{\rho}_{lk}[n] = \sqrt{1 - \rho_{lk}^{2}[n]}$ and $\z_{lk}^{j}[n]$ represents the innovation due to channel aging. Under the i.i.d Rayleigh fading model, both $\g_{lk}^{j}[n]$ and $\z_{lk}^{j}[n]$ conform to $\CN\left(0, \beta_{lk}^{j}\I_{M_c}\right)$ distribution with $\beta_{lk}^{j}$ denoting the large-scale fading coefficient between BS $j$ and $\UE_{lk}$. Then, we have the following two propositions:

\begin{proposition}
With downlink training, the $n$th-instant approximate achievable downlink SE of $\UE_{lk}$ in the cellular mMIMO network described above is
\begin{align}\label{eq:CELL_SE_DT}
    \SE_{lk}^{\text{cell},\DT}[n] = \log_{2}\left(1 + \text{SINR}_{lk}^{\text{cell},\DT}[n]\right),
\end{align}
where $\text{SINR}_{lk}^{\text{cell},\DT}[n]$ denotes the $n$th-instant effective downlink SINR of $\UE_{lk}$ in the presence of downlink training, given by
\begin{align}\label{eq:CELL_SINR_DT}
    \text{SINR}_{lk}^{\text{cell},\DT}[n] = \frac{\rho_{lk}^{2}[n]\left(M_{c}\sqrt{\eta_{lk}}\gamma_{lk}^{l}\right)^{2} + \rho_{lk}^{2}[n]\kappa_{lk}}{\sum_{l'=1}^{L_c}\sum_{i=1}^{K_c}M_{c}\eta_{l'i}\gamma_{l'i}^{l'}\beta_{lk}^{l'} + \rho_{lk}^{2}[n]\sum_{\substack{l'=1\\ l' \neq l}}^{L_c}\left(M_{c}\sqrt{\eta_{l'k}}\gamma_{l'k}^{l'}\frac{\beta_{lk}^{l'}}{\beta_{l'k}^{l'}}\right)^{2} - \rho_{lk}^{2}[n]\kappa_{lk} +   \frac{1}{\Ed}},
\end{align}
where $\kappa_{lk} = \frac{\Tdp\Edp\left(M_{c}\eta_{lk}\gamma_{lk}^{l}\beta_{lk}^{l}\right)^{2}}{1 + \Tdp\Edp\sum_{l'=1}^{L_c}M_{c}\eta_{l'k}\gamma_{l'k}^{l'}\beta_{lk}^{l'}}$ with $\eta_{lk}$ denoting the power control coefficient of $\UE_{lk}$, $\gamma_{lk}^{l}$ the variance of the uplink channel estimate between $\UE_{lk}$ and the $l$th BS, $\Tdp$ the length of the downlink pilots and $\Edp$ the normalized downlink transmit SNR.
\begin{proof}
See Appendix \ref{app:proof_proposition_2}.
\end{proof}
\end{proposition}

\begin{proposition}
When the UEs rely on statistical CSI for data decoding, the $n$th-instant approximate achievable downlink SE of $\UE_{lk}$ in cellular mMIMO is
\begin{align}\label{eq:CELL_SE_CSI}
    \SE_{lk}^{\text{cell},\CSI}[n] = \log_{2}\left(1 + \text{SINR}_{lk}^{\text{cell},\CSI}[n]\right)
\end{align}
where $\text{SINR}_{lk}^{\text{cell},\CSI}[n]$ denotes the $n$th-instant effective downlink SINR of $\UE_{lk}$ relying on statistical CSI, given by
\begin{align}\label{eq:CELL_SINR_CSI}
    \text{SINR}_{lk}^{\text{cell},\CSI}[n] = \frac{\rho_{lk}^{2}[n]\left(M_{c}\sqrt{\eta_{lk}}\gamma_{lk}^{l}\right)^{2}}{\sum_{l'=1}^{L_c}\sum_{i=1}^{K_c}M_{c}\eta_{l'i}\gamma_{l'i}^{l'}\beta_{lk}^{l'} + \rho_{lk}^{2}[n]\sum_{\substack{l'=1\\ l' \neq l}}^{L_c}\left(M_{c}\sqrt{\eta_{l'k}}\gamma_{l'k}^{l'}\frac{\beta_{lk}^{l'}}{\beta_{l'k}^{l'}}\right)^{2} + \frac{1}{\Ed}}.
\end{align}
\begin{proof}
See Appendix \ref{app:proof_proposition_3}.
\end{proof}
\end{proposition}

As in cell-free mMIMO, the difference between the  SINRs with and without downlink training in \eqref{eq:CELL_SINR_DT} and \eqref{eq:CELL_SINR_CSI} for cellular mMIMO, is an addition and subtraction of the term $\rho_{lk}^{2}[n]\kappa_{lk}$ in the numerator and denominator of \eqref{eq:CELL_SINR_CSI}, respectively. This suggests that the effect of downlink training on the performance of a UE is similar in both cell-free and cellular mMIMO networks, namely an increase in the instantaneous downlink SE. Furthermore, this gain reduces at higher transmission indices or higher mobility due to channel aging. However, the relative gain in performance over the statistical CSI scheme may be different in the two networks owing to the different degrees of channel hardening. For reasons mentioned in Section \ref{sec:introduction}, it is expected that the boost in the downlink SE due to downlink training may not be as high in cellular mMIMO as in cell-free mMIMO.

Under an equal number of AP/BS antennas in both cell-free and cellular networks, we expect that the average downlink SE in cell-free mMIMO will be substantially higher than that in cellular mMIMO at all levels of mobility with both downlink training and statistical CSI, owing to the distributed processing enabled by cell-free mMIMO.

Finally, note that, upon substituting $n = 0$, the expressions in \eqref{eq:CELL_SE_DT} and \eqref{eq:CELL_SE_CSI} reduce to those given in \cite{zuo_tvt_2016} where the effect of a time-varying channel is ignored.

\section{Numerical Results}\label{sec:numerical_results}

In this section, we provide numerical results that demonstrate the effect of user mobility on the downlink performance of cell-free and cellular mMIMO. We begin by describing the simulation setup. Then, we present the performance of the two networks.

\subsection{Simulation Setup}
We consider the propagation model proposed in \cite{bjornson_twc_2020}. Inside a $1\,\textrm{km}^{2}$ square region, for the cell-free setup, there are $M = 100$ APs each equipped with $L = 4$ antennas that are placed at points on a uniform grid. The APs serve $K = 40$ UEs that are uniformly distributed at random locations in the region. For the cellular setup, there are $L_c = 4$ BSs each equipped with $M_c = 100$ antennas which serve $K_c = 10$ UEs per cell. The simulation involves computing the average downlink SE of a UE and the sum-SE of the network for a given system realization and repeating the same for 400 realizations. Since the location of the UEs is random in each instance, we focus on the $90\%$-likely average downlink SE \cite{bjornson_twc_2020} and the mean of the sum-SE across all such instances.

The channels across multiple antennas of the AP/BS are spatially uncorrelated. We consider carrier frequency $2~\text{GHz}$, bandwidth $1~\text{MHz}$, noise figure $9~\text{dB}$, AP transmit power $200\text{mW}$ and UE transmit power $100\text{mW}$. The uplink/downlink pilot assignment for the cell-free setup is performed as per \cite[Algorithm 2]{interdonato_twc_2019} with orthogonal pilot reuse. Unless stated otherwise, we assume $\Tup = \Tdp = 10$ symbols in all the figures. The APs (BSs) transmit at full power and allocate uniformly among the $K$ ($K_c$) UEs that they serve. Thus, the power control coefficients are set as $\eta_{mk} = \left(L\sum_{k'=1}^{K}\gamma_{mk'}\right)^{-1}$ for the cell-free setup and $\eta_{ji} = \left(M_{c}\sum_{i'=1}^{K_c}\gamma_{ji'}^{j}\right)^{-1}$ for the cellular setup. All UEs are assumed to move at the maximum possible velocity $\Vmax$.

\subsection{Cell-Free mMIMO}

Fig.~\ref{fig:mostProbableRateCF_1} shows the variation in the $90\%$-likely average downlink SE with $V_{\textrm{max}}$ for different data duration. We plot the performance with no pilot contamination (\texttt{NoPC}) where $\tau_{\text{up}} = \tau_{\text{dp}} = K = 40$ symbols, and with pilot contamination (\texttt{PC}) where  $\tau_{\text{up}} = \tau_{\text{dp}} = 10$ symbols, for both the downlink training (\texttt{DT}) and the statistical CSI (\texttt{StCSI}) schemes. Thus, [\texttt{NoPC DT} $\Tdd=200$] incurs the highest relative training overhead and [\texttt{PC StCSI} $\Tdd = 500$] incurs the least. In all cases, the average SE decreases as $\Vmax$ increases, in line with our analytical results. At low user mobility, the effect of channel aging is small. Consequently, there is little loss in the downlink SE across transmission indices and the AP can obtain a higher average SE by transmitting more data symbols in a transmit frame. Thus, at $\Vmax = 5~\textrm{m/s}$, setting $\Tdd = 500$ symbols yields the higher average SE for all four sets of curves. Overall, it is worthwhile to obtain better quality channel estimates by using a longer pilot duration, avoiding pilot contamination and employing downlink training. Thus, [\texttt{noPC} \texttt{DT} $\Tdd = 500$] yields the highest average SE at low user mobility. This is followed by [\texttt{NoPC StCSI} $\Tdd = 500$] and [\texttt{PC DT} $\Tdd = 500$] which offer nearly the same performance by either avoiding pilot contamination and using statistical CSI or incurring pilot contamination and using downlink training. As expected, with pilot contamination, statistical CSI and shorter data duration ($\Tdd = 200$), the UEs achieve the worst performance. On the other hand, as $\Vmax$ increases, the downlink SE drops so quickly with the symbol index that more frequent re-estimation of the channel is necessary. Hence, in the high mobility regime ($\Vmax = 85$~m/s), it is better to transmit fewer data symbols in the transmit frame, i.e., $\Tdd = 200$ outperforms $\Tdd = 500$. Moreover, since avoiding pilot contamination and using downlink training both incur additional overhead, the use of both techniques is not beneficial at high mobility: [\texttt{PC DT} $\Tdd = 200$] offers the highest SE followed by [\texttt{noPC StCSI} $\Tdd = 200$] followed by [\texttt{noPC DT} $\Tdd = 200$] with [\texttt{PC StCSI} $\Tdd = 200$] offering the least SE among the four curves. Nonetheless, we see that the use of downlink training is important: up to $\sim\Vmax = 55$~m/s, [\texttt{noPC DT} $\Tdd = 200$] yields the highest average SE after which [\texttt{PC DT} $\Tdd = 200$] performs the best. Thus, overall, it is better to use downlink training than statistical CSI-based decoding in cell-free mMIMO.

\begin{figure}
    \centering
    \includegraphics[width=4.5in]{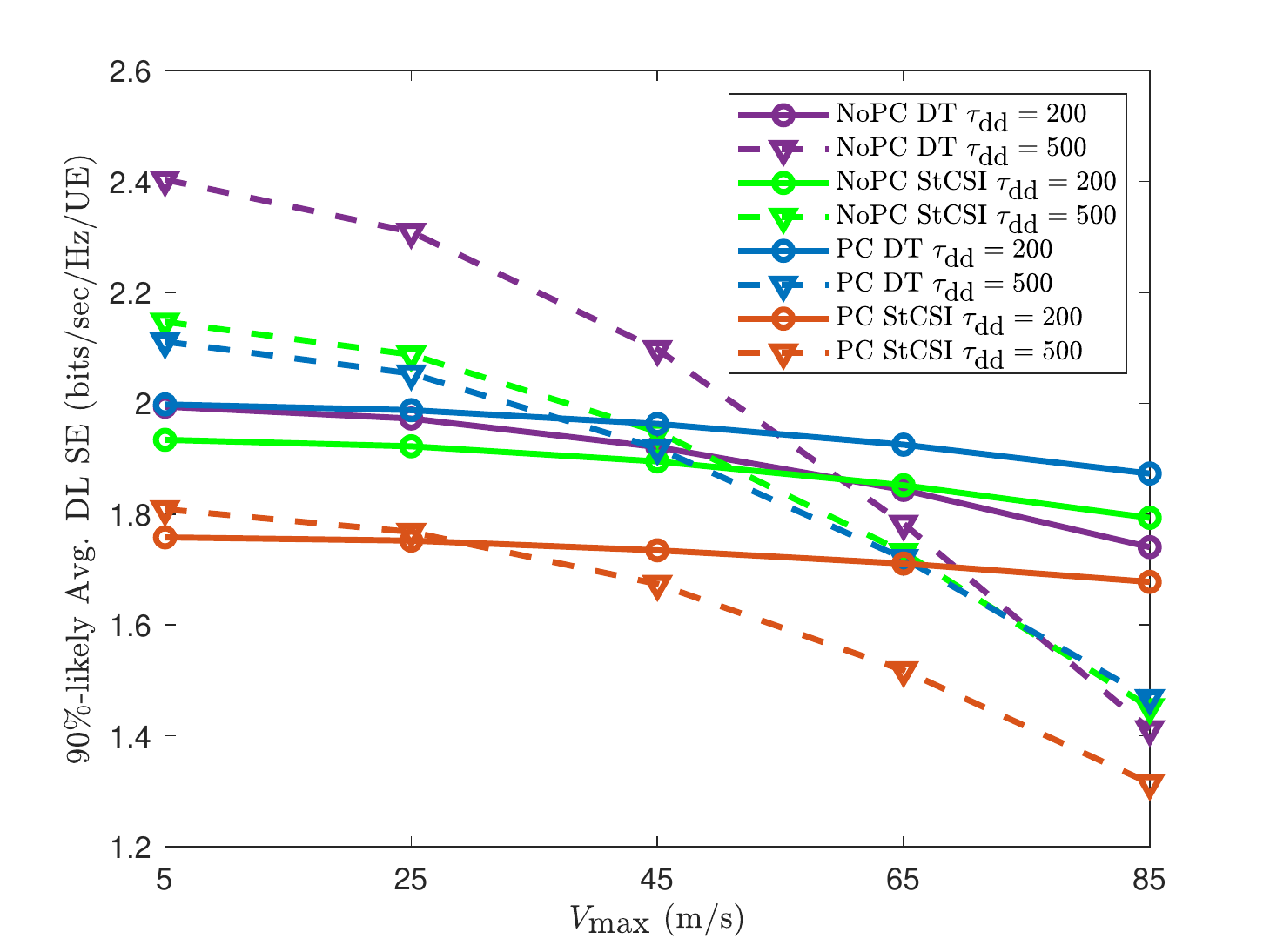}
    \caption{The $90\%$-likely average downlink SE  as a function of $V_{\textrm{max}}$. NoPC represents no pilot contamination as in $\Tup=\Tdp=40$ symbols for the downlink training scheme and $\Tdp = 0$ for the statistical CSI scheme, respectively. The performance gain due to downlink training reduces at higher user mobility.}
    \label{fig:mostProbableRateCF_1}
\end{figure}

Fig.~\ref{fig:mostProbableRateCF_2} shows a plot of the $90\%$-likely average downlink SE as a function of the data duration when $\Tup = \Tdp = 10$ symbols. When the number of data symbols in a frame is comparable to the length of the training interval (i.e. $10 \leq \Tdd \leq 50$ symbols), the statistical CSI scheme yields better average SE than the downlink training scheme. This is because for small data duration, the difference between the summation in \eqref{eq:AV_SE_DT} and \eqref{eq:AV_SE_CSI} is marginal and in such cases, the pre-sum factor dominates. However, the average SE improves with increasing $\Tdd$, and ultimately, the downlink training scheme ends up performing better. Furthermore, the average SE is a unimodal function of the data duration and it is maximized at a $\Tdd$ of around $150, 300$ and over $600$ symbols for $\Vmax = 85, 45$ and $5$~m/s, respectively. This behavior corroborates that of the analytical expression for the average SE in \eqref{eq:AV_SE_DT} and \eqref{eq:AV_SE_CSI} and is because the downlink SE defined in \eqref{eq:CF_SE_DT} and \eqref{eq:CF_SE_CSI} is monotonically decreasing with~$n$. When the system is operated at the optimal value of $\Tdd$, the downlink training scheme outperforms the statistical CSI scheme, achieving about $15\%$ better~SE.

\begin{figure}
    \centering
    \includegraphics[width=4.5in]{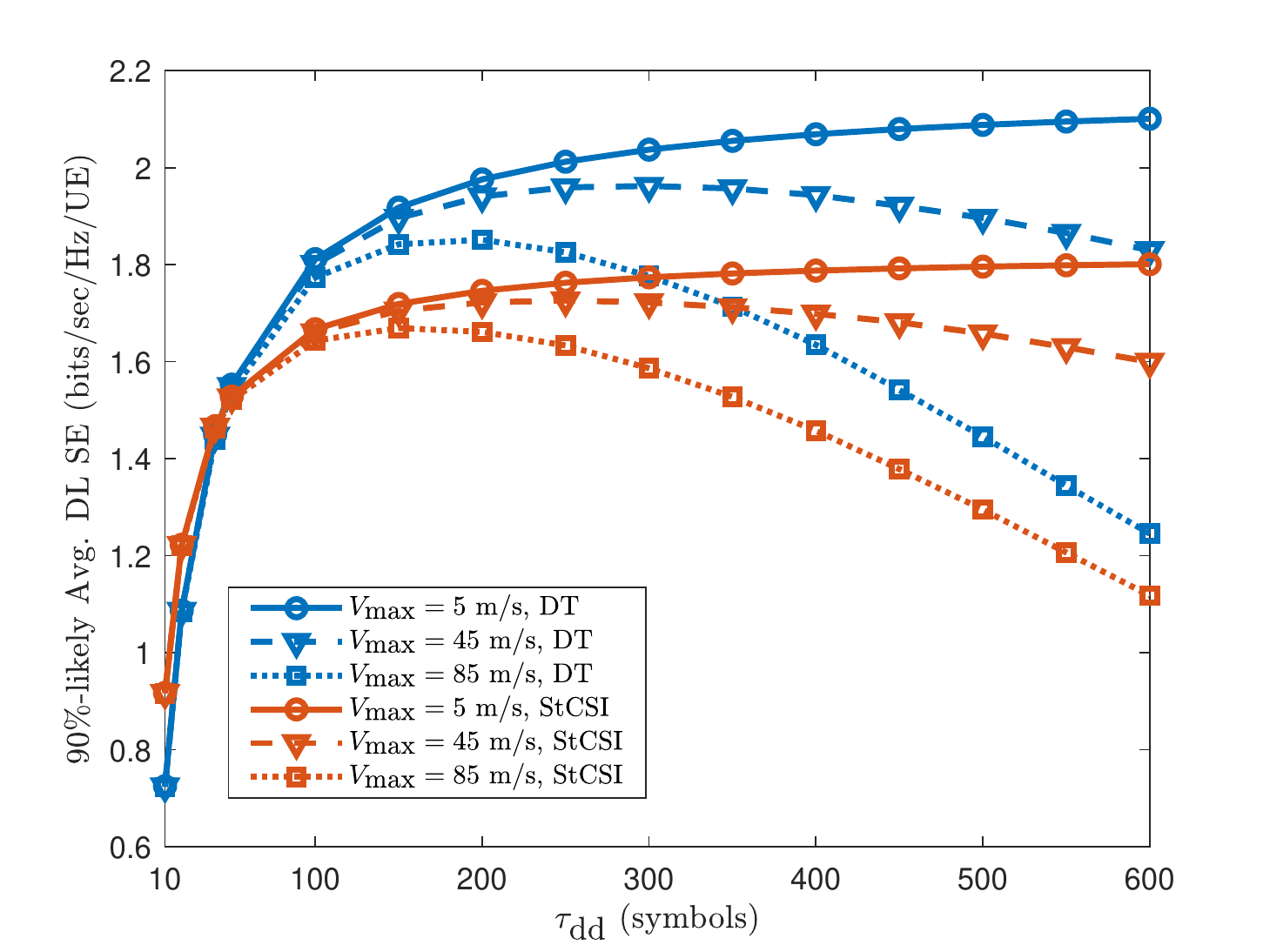}
    \caption{Plot of the $90\%$-likely average downlink SE versus the data duration $\Tdd$. Statistical CSI scheme outperforms the downlink training scheme for small data percentages.}
    \label{fig:mostProbableRateCF_2}
\end{figure}

In Figure \ref{fig:mostProbableRateCF_2_comparePC}, we plot the $90\%$-likely average downlink SE with downlink training (fig. \ref{fig:mostProbableRateCF_comparePC_2_DT}) and statistical CSI (fig. \ref{fig:mostProbableRateCF_comparePC_2_stCSI}) as a function of the data duration. In both cases, we contrast the performance obtained with no pilot contamination against that with pilot contamination. With downlink training, at low mobility, it is better to use a longer data duration and avoid pilot contamination. In contrast, at high mobility, the channel estimation overhead comes at a premium, so it is better to use a shorter data duration ($\sim 150$ symbols) and a shorter pilot length ($10$ symbols) even though it incurs pilot contamination. This holds true at low mobility even with statistical CSI, but at high mobility, it is important to obtain good initial channel estimates at the APs, and hence we see that the no pilot contamination scheme with $\Tup = 40$ symbols offers the best performance at $\Tdd = 200$ symbols. Similar to the previous figure, we see that when the data duration is chosen to maximize the average SE, the downlink training outperforms the statistical CSI scheme.

\begin{figure}
    \centering
    \subfigure[]
    {
        \includegraphics[width=3.0in]{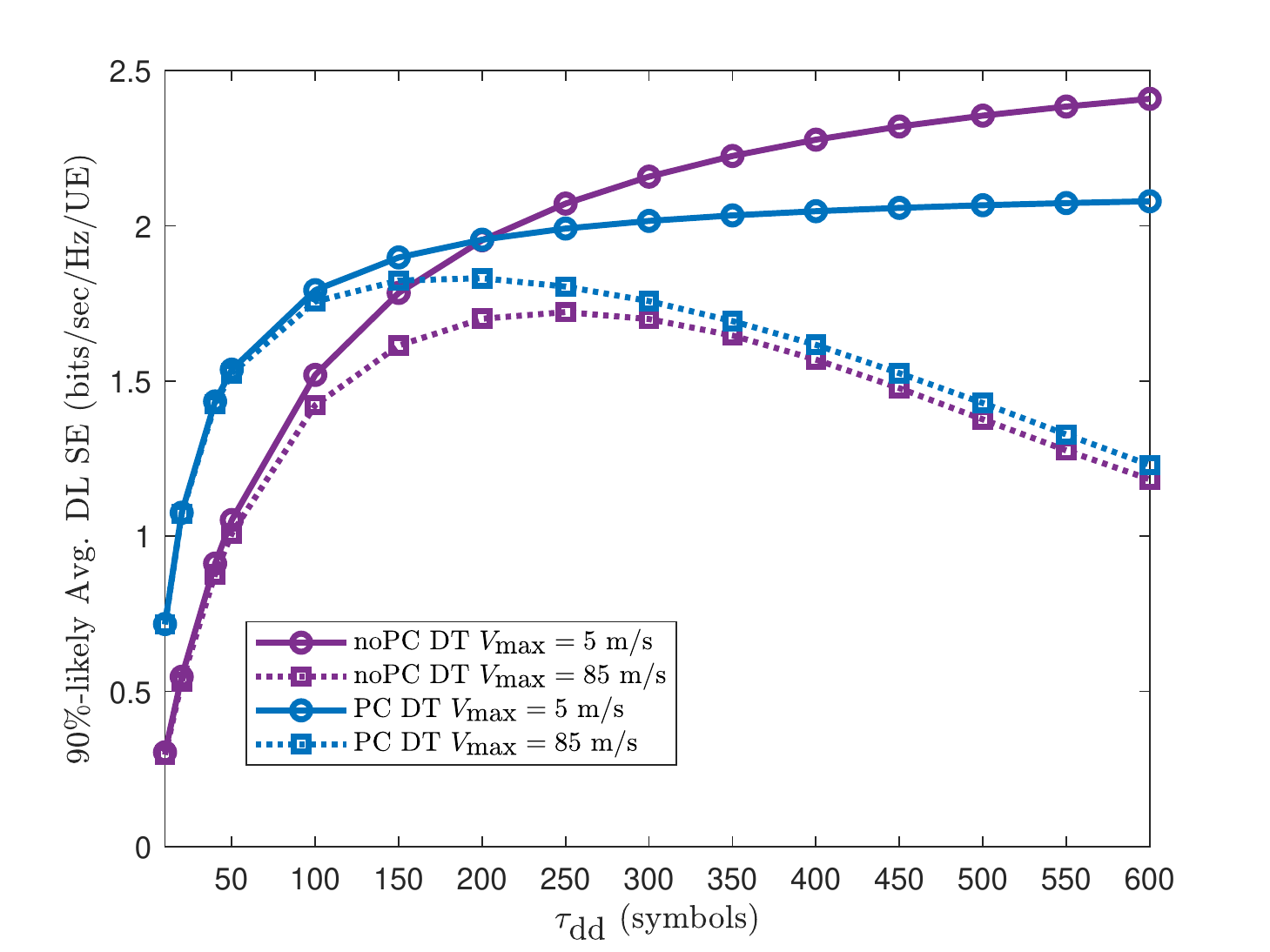}
        \label{fig:mostProbableRateCF_comparePC_2_DT}
    }
    \subfigure[]
    {
        \includegraphics[width=3.0in]{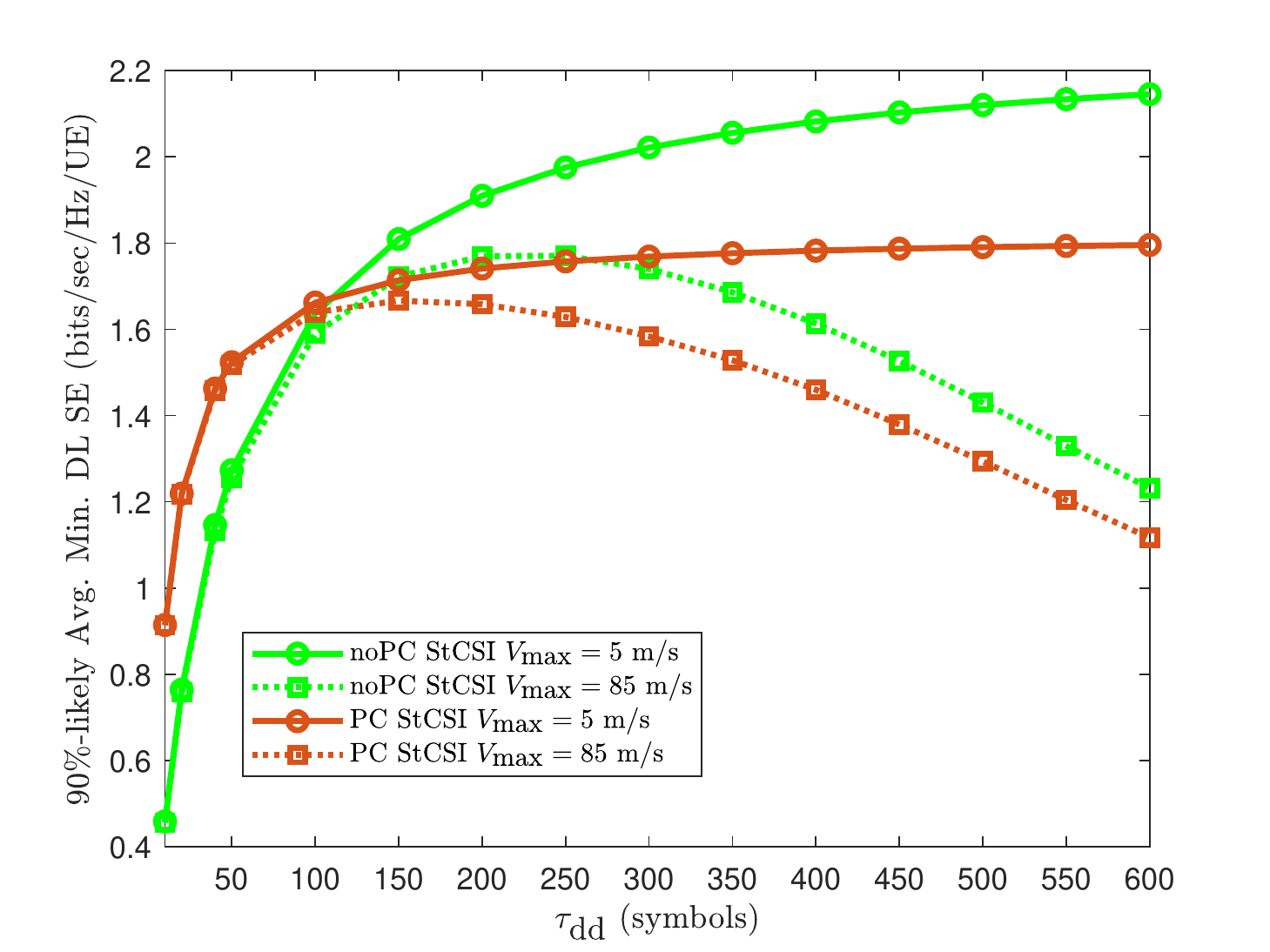}
        \label{fig:mostProbableRateCF_comparePC_2_stCSI}
    }
    \caption{Plots of the $90\%$-likely average downlink SE with and without downlink training.}
    \label{fig:mostProbableRateCF_2_comparePC}
\end{figure}

Next, we focus on the impact of the relative uplink/downlink pilot lengths on the downlink performance of a UE. Figure \ref{fig:mostProbableRateCF_3} shows a plot of the $90\%$-likely average downlink SE with $\Vmax$ for different combinations of the uplink and downlink pilot lengths. The total length of the training interval is kept fixed across all curves, equal to 30 symbols. In such a case, a higher $\Tup$ would correspond to more number of available pilots on the uplink and fewer pilots on the downlink which in turn would imply lesser pilot contamination on the uplink and more pilot contamination on the downlink. From the figure, it can be seen that the average SE for a given $\Vmax$ tends to increase as $\Tup$ increases. This means that uplink pilot contamination exerts more control on the downlink performance of a UE than downlink pilot contamination does. The reason is that channel estimates impaired by uplink pilot contamination pass on the imperfectness to the subsequent downlink training and data transmission stages via the transmit precoding step. Although this phenomenon was brought up in \cite{interdonato_twc_2019}, figure \ref{fig:mostProbableRateCF_3} conveys an important corollary to it: the relative gain in the average downlink SE due to a finite increment in $\Tup$ (roughly) reduces as the pilot length approaches its upper limit (equal to 30 symbols) even when the increment is kept fixed (equal to 5 symbols).

\begin{figure}
    \centering
    \includegraphics[width=4.5in]{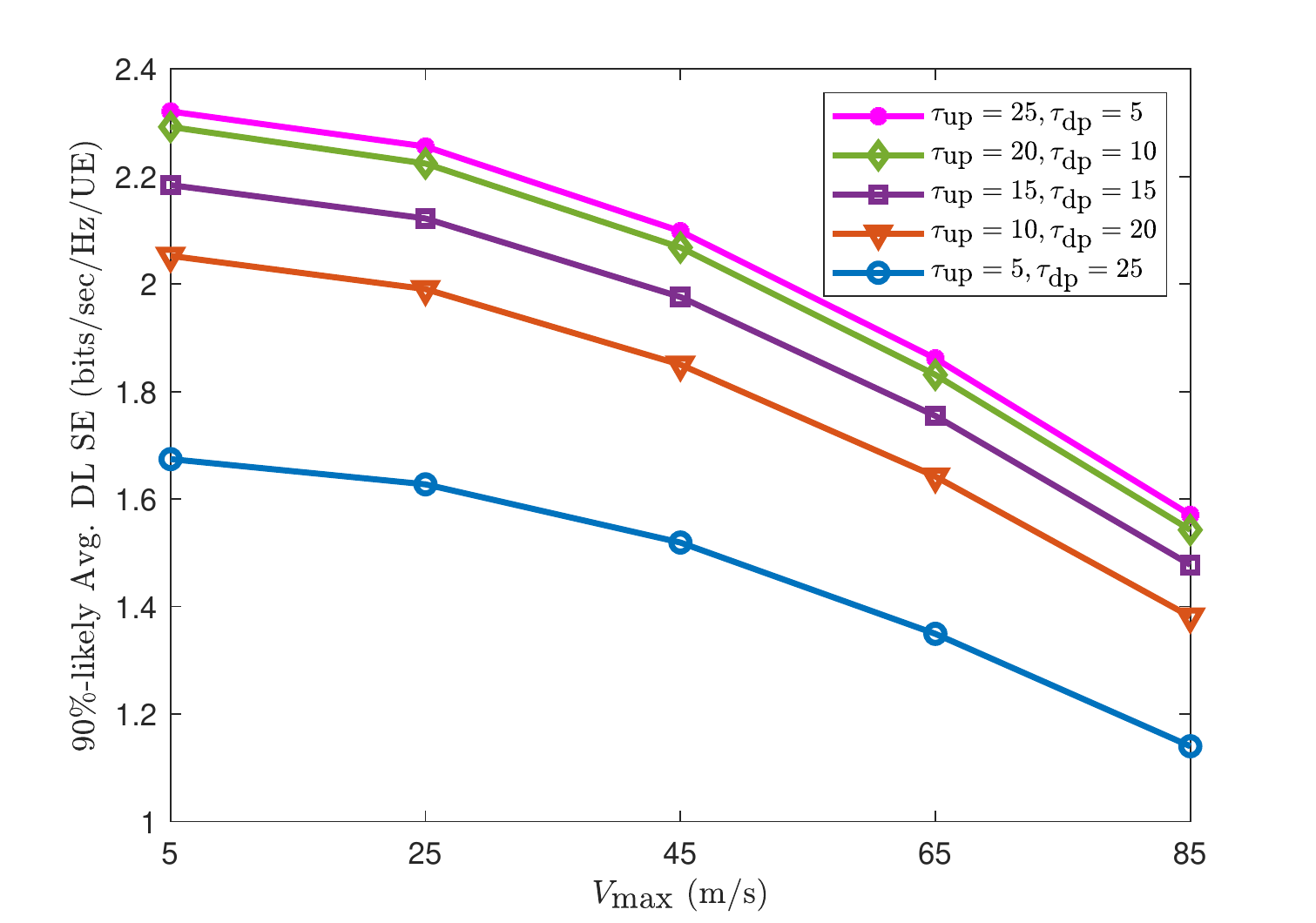}
    \caption{Plot of the $90\%$-likely average downlink SE vs. $V_{\textrm{max}}$ for different $\{\Tup,\Tdp\}$ pairs when $\Tdd = 500$ symbols. The relative gain in performance from having longer uplink pilots reduces as the pilot length approaches its upper limit.}
    \label{fig:mostProbableRateCF_3}
\end{figure}
 
Next, we focus on the effect of densifying a cell-free network with APs when the total number of AP antennas in the network is held fixed. Figure \ref{fig:channelHardeningCF} shows a plot of the cumulative distribution function (CDF) of the average downlink SE of a UE at $\Vmax = 5~\textrm{m/s}$ with and without downlink training. We consider two deployments: one involving 100 APs with 4 antennas each and another involving 400 APs with a single antenna each. As more APs are added in the network, the average distance between a UE and an AP reduces and it is expected that the downlink performance of a UE improve. From the figure, we observe that densification improves the $90\%$-likely average SE significantly provided the UEs receive downlink training. When UEs rely on statistical CSI, however, the $90\%$-likely average SE marginally drops as more APs are added in the network. This is because the performance of the statistical CSI scheme depends on how close the actual value of the downlink channel is to the mean value, which in turn is determined by the amount of channel hardening. Channel hardening in cell-free mMIMO is governed primarily by the APs that are located geographically close to a UE. When the number of antennas on an AP drops from 4 to 1, there is loss in the downlink SE owing to less channel hardening which counteracts any improvement arising from the reduced AP-UE distance.

\begin{figure}
    \centering
    \includegraphics[width=4.5in]{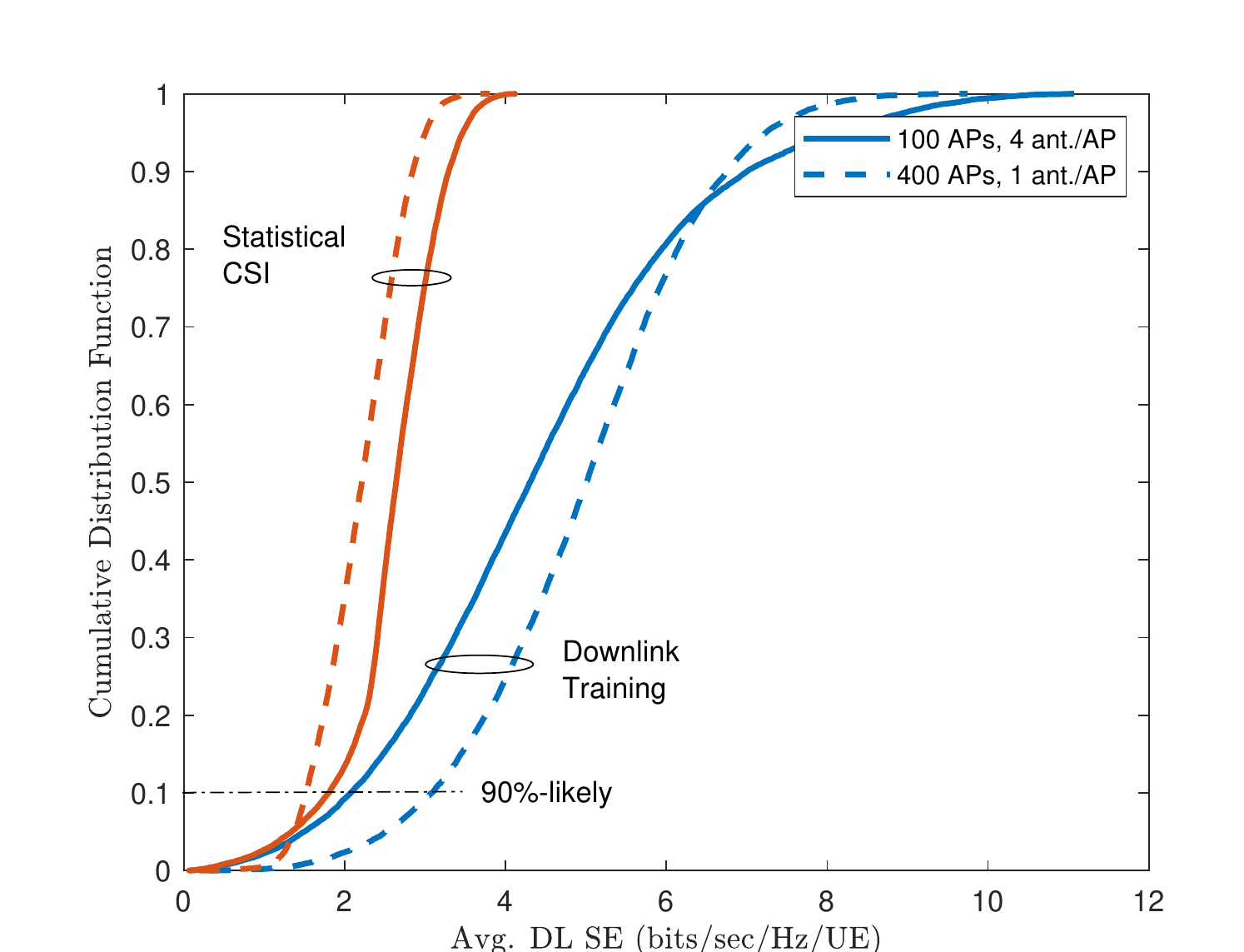}
    \caption{CDF plot of the average downlink SE at $\Tdd = 500$ symbols and $\Vmax = 5~\textrm{m/s}$. Densification induces significant gain in the average SE provided the UEs receive downlink pilots.}
    \label{fig:channelHardeningCF}
\end{figure}

Figure \ref{fig:sumRateCF_4} shows a plot of the average sum-SE of the cell-free network drawn as a function of $\Vmax$ for the two scenarios presented in Fig.~\ref{fig:channelHardeningCF}. When UEs rely on statistical CSI, the sum-SE obtained from having fewer APs with more antennas each is significantly higher than a dense AP deployment, and this gain in performance is observed across all levels of mobility regardless of pilot contamination. With downlink training, having more APs with fewer antennas each yields better sum-SE. However, this gain becomes negligible (or even slightly negative depending on the severity of pilot contamination) at extreme user mobility. This is because, at high UE velocities, the downlink channel estimates get outdated rapidly, and as a consequence, the advantage due to the lower UE-AP distance reduces. 

\begin{figure}
    \centering
    \includegraphics[width=4.5in]{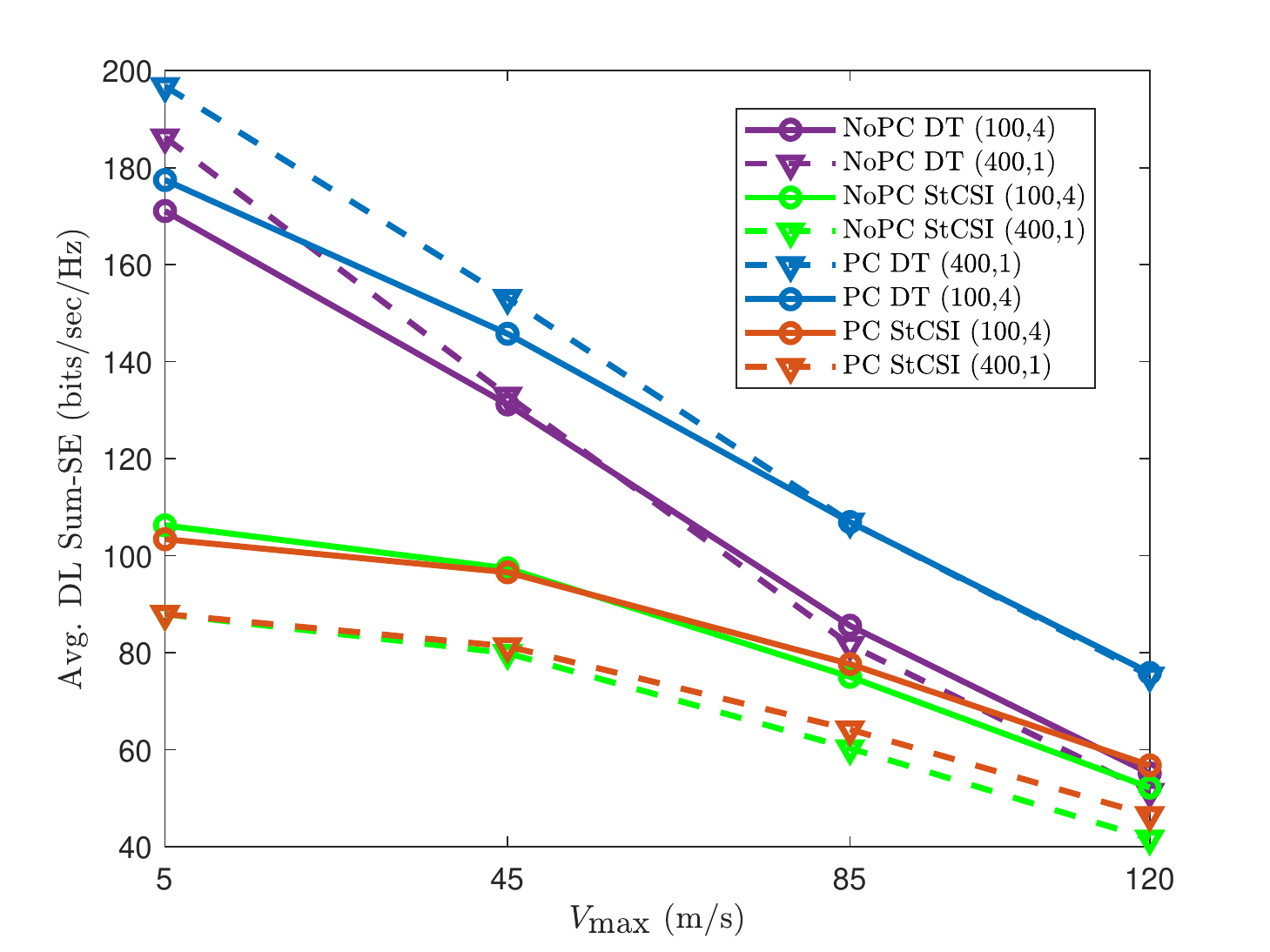}
    \caption{Plot of the network sum-SE vs. $\Vmax$ when $\Tdd = 500$ symbols. With downlink training, the gain in the sum-SE due to densification diminishes as the UEs move faster.}
    \label{fig:sumRateCF_4}
\end{figure}

\subsection{Cellular mMIMO}
Figure \ref{fig:cellular1} shows plots of the $90\%$-likely average downlink SE and the average sum-SE in cellular mMIMO drawn as a function of $\Vmax$. Although the effect of mobility remains largely the same as in cell-free mMIMO, in cellular mMIMO, both downlink training and statistical CSI schemes perform equally well (in terms of the $90\%$-likely measure) with the latter marginally outperforming the former. Owing to the high degree of channel hardening and the smaller training overhead, the statistical CSI scheme yields slightly better $90\%$-likely average downlink SE than the downlink training scheme when we account for the time-varying channel. However, the average sum-SE is found to be higher with the downlink training scheme. Thus, although it does not improve the $90\%$-likely SE by much, downlink training is helpful in improving the average sum-SE of cellular mMIMO. This is because the improved channel estimates obtained via downlink training allow the best-performing UEs (UEs with high SEs) to retain their SE for a longer duration, thus improving the sum-SE.

\begin{figure}
    \centering
    \subfigure[]
    {
        \includegraphics[width=3.0in]{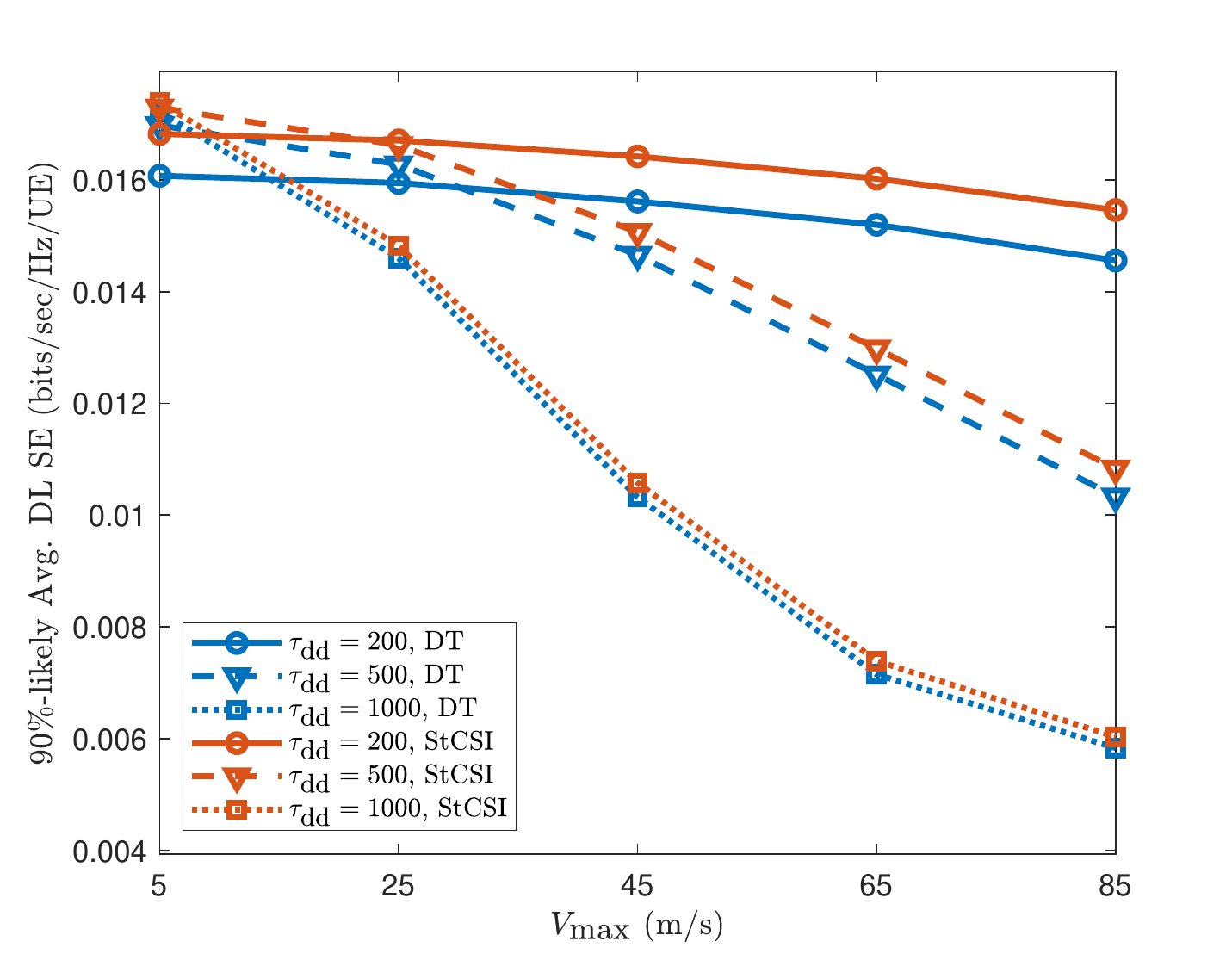}
        \label{fig:mostProbableRateCellular_1}
    }
    \subfigure[]
    {
        \includegraphics[width=3.0in]{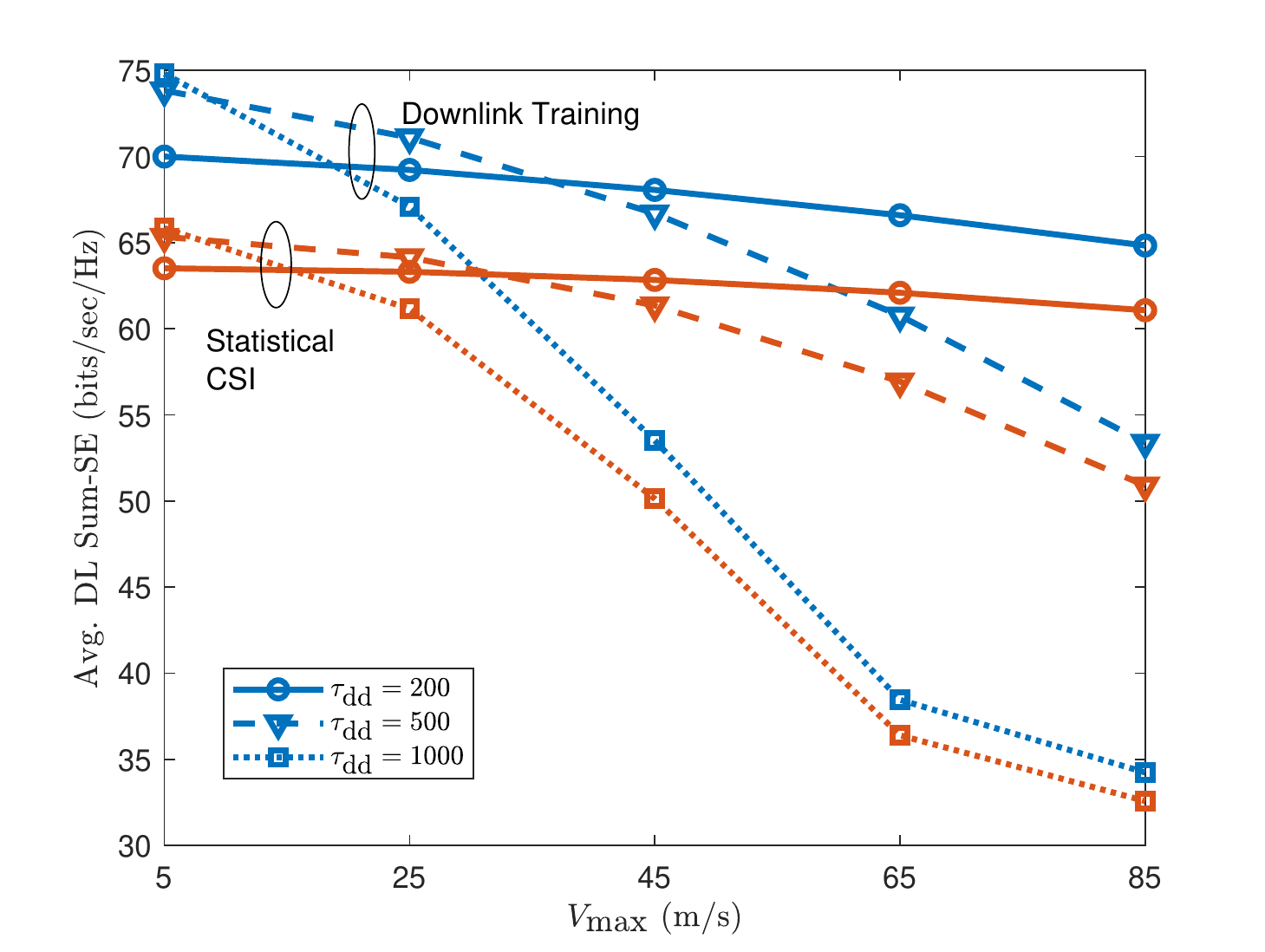}
        \label{fig:sumRateCellular_1.eps}
    }
    \caption{Plots of (a) $90\%$-likely average downlink SE and (b) average downlink sum-SE drawn as a function of $\Vmax$ in a cellular mMIMO network.}
    \label{fig:cellular1}
\end{figure}

Figure \ref{fig:cellular2} shows plots of the $90\%$-likely average downlink SE and the average sum-SE drawn as a function of the data duration $\Tdd$. As in cell-free mMIMO, the average downlink SE in cellular mMIMO is found to possess a non-monotonic behaviour with respect to $\Tdd$. However, comparing figures \ref{fig:mostProbableRateCellular_2} and \ref{fig:mostProbableRateCF_2}, we find that the relative behavior between the downlink training and the statistical CSI scheme is different in both networks. Unlike cell-free mMIMO, in cellular mMIMO, the disparity between the two schemes remains consistent across the entire $\Tdd$ values. Further, the performance gap between the two schemes is maximum initially and tends to reduce as data duration increases. With cell-free mMIMO, however, the two curves tend to start off together (after the cross-over) at low $\Tdd$ values and diverge for longer data sequences. From a sum-SE perspective, however, it is found that for smaller data duration, the statistical CSI scheme yields better sum-SE performance in cellular mMIMO, but as the data duration increases, downlink training scheme tends to outperform. The key takeaway is that in cellular mMIMO, depending on what is to be optimized (user-fairness or sum-SE), either of downlink training and statistical CSI may be  preferred.

\begin{figure}
    \centering
    \subfigure[]
    {
        \includegraphics[width=3.0in]{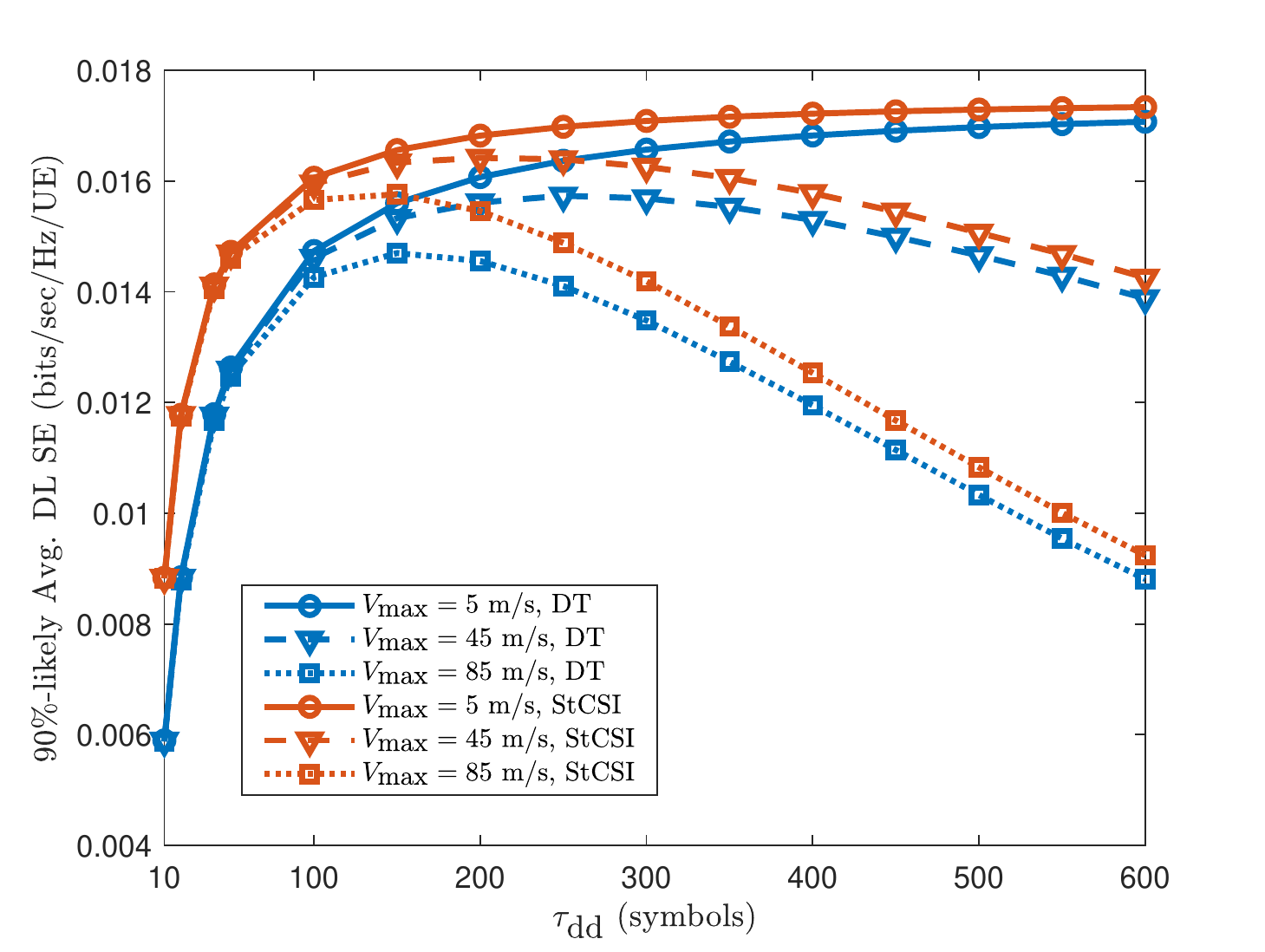}
        \label{fig:mostProbableRateCellular_2}
    }
    \subfigure[]
    {
        \includegraphics[width=3.0in]{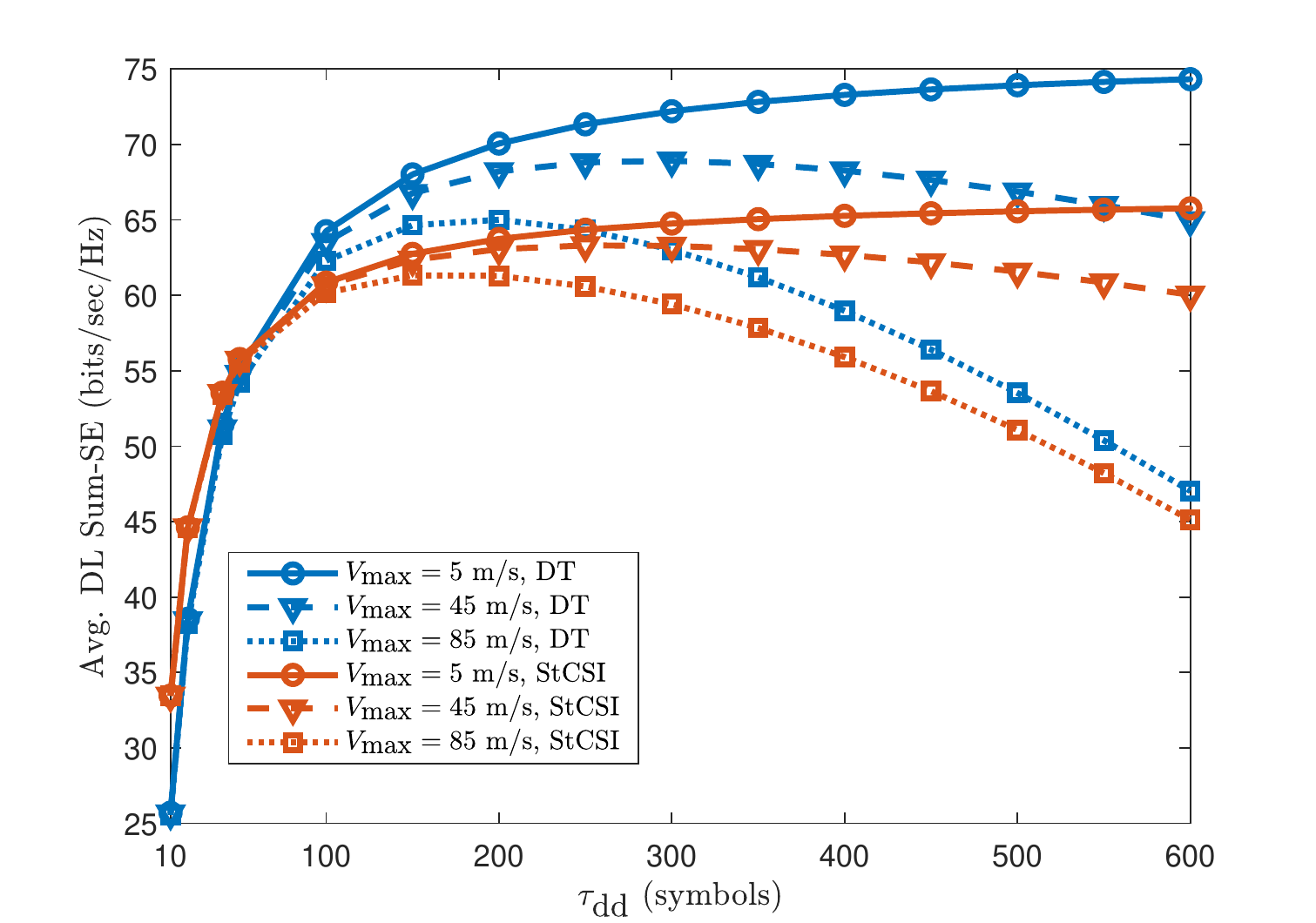}
        \label{fig:sumRateCellular_2.eps}
    }
    \caption{Plots of (a) $90\%$-likely average downlink SE and (b) average downlink sum-SE drawn as a function of $\Tdd$ in a cellular mMIMO network.}
    \label{fig:cellular2}
\end{figure}

\section{Conclusions}\label{sec:conclusion}

In this paper, focusing on the downlink, we analyzed the performance of a cell-free mMIMO network while accounting for user mobility. We showed that mobility results in multi-user interference in the form of channel aging and pilot contamination which degrades the performance of the UEs present in the network. Using numerical results, the effects of the maximum user velocity, the data duration and the uplink/downlink training lengths on the per-user SE and the sum-SE were illustrated. While downlink training is beneficial to cell-free mMIMO UEs, the gain in performance depends on factors such as user mobility and the data duration in the transmit frame. Furthermore, when the total number of AP antennas in the network is fixed, it is far better to have more APs with fewer antennas each than otherwise. When considering the sum-SE, however, this gain in performance reduces at high user mobility. Finally, in a cellular mMIMO network, UEs generally achieve slightly better 90\% likely SE when they rely on statistical CSI to decode downlink data. However, if the sum-SE is to be optimized, downlink training must be preferred. This difference arises from the different degrees of channel hardening in cell-free and cellular mMIMO.  Future work can consider the impact of spatial correlation in the channels.

\appendices

\section{Proof of Proposition 1}
\label{app:proof_proposition_1}
Recall from \eqref{eq:DLDataReceive2} that the $n$th received data signal at the $k$th UE is given by
\begin{align}
    r_{\text{d}, k}[n] =  \sqrt{\Ed}d_{kk}[n] q_{k}[n] + \Tilde{w}_{\text{d},k}[n],
\end{align}
where $\Tilde{w}_{\text{d},k}[n] \triangleq \sqrt{\Ed} \sum_{k' \neq k}^{K} d_{kk'}[n] q_{k'}[n] + w_{\text{d},k}[n]$ represents the effective non-Gaussian noise. Assuming $\{q_{k'}[n]\}$ has zero mean and is independent of $d_{kk'}[n]~\forall~k,k'$, we can write 
\begin{align}
    \expect\{\Tilde{w}_{\text{d},k}[n]|\hat{d}_{kk}\} = \expect\{q_{k}^{*}[n]\Tilde{w}_{\text{d},k}[n]|\hat{d}_{kk}\} = \expect\{d_{kk}^{*}[n]q_{k}^{*}[n]\Tilde{w}_{\text{d},k}[n]|\hat{d}_{kk}\} = 0.
\end{align}
Now, the achievable downlink SE of the $k$th UE can be computed using the capacity bounding technique in \cite{marzetta_book} for a fading channel with non-Gaussian noise where the receiver has access to side information as
\begin{align}\label{eq:capacity0}
    \SE_{k}[n] \geq \expect\left(\text{log}_{2}\left(1 + \frac{\Ed \Big|\expect\Big\{d_{kk}[n]\Big|\hat{d}_{kk}\Big\}\Big|^{2}}{\Ed\sum_{k'=1}^{K}\expect\Big\{|d_{kk'}[n]|^{2}\Big|\hat{d}_{kk}\Big\} - \Ed \Big|\expect\Big\{d_{kk}[n]\Big|\hat{d}_{kk}\Big\}\Big|^{2}+1}\right)\right).
\end{align}
From~\eqref{eq:d_kk'[n]}, $d_{kk}[n]$ can be expressed as
\begin{align}
    d_{kk}[n] = \rho_{k}[n]\hat{d}_{kk} + \rho_{k}[n]\Tilde{d}_{kk} + \bar{\rho}_{k}[n]z_{kk}[n],
\end{align}
where the quantities $\hat{d}_{kk}, \Tilde{d}_{kk}$ and $z_{kk}[n]$ are mutually uncorrelated with the latter two terms having zero mean. In section II, the quantity $d_{kk}[n]$ was shown to be approximately Gaussian for finite values of $M$ and $L$. Furthermore, the quantities $\hat{d}_{kk}, \Tilde{d}_{kk}$ and $z_{kk}[n]$ are (approximately) jointly Gaussian and hence are statistically independent too. This allows us to write
\begin{align}\label{eq:SINR_Numerator_Final}
    \expect\Big\{d_{kk}[n]\Big|\hat{d}_{kk}\Big\} = \rho_{k}[n]\hat{d}_{kk}.
\end{align}
Now, focusing on the denominator term of the SINR in \eqref{eq:capacity0}, we can write 
\begin{align}
    \Ed\sum_{k'=1}^{K}\expect\Big\{|d_{kk'}[n]&|^{2}\Big|\hat{d}_{kk}\Big\} - \Ed \Big|\expect\Big\{d_{kk}[n]\Big|\hat{d}_{kk}\Big\}\Big|^{2}+1 \nonumber \\
    = \Ed\sum_{k'\neq k}^{K}&\expect\Big\{|d_{kk'}[n]|^{2}\Big|\hat{d}_{kk}\Big\} + \Ed\expect\Big\{|d_{kk}[n]|^{2}\Big|\hat{d}_{kk}\Big\} - \Ed \rho_{k}^{2}[n]|\hat{d}_{kk}|^{2}  + 1, \label{eq:SINR_Denominator}
\end{align}
where
\begin{align}
    \expect\{|d_{kk}[n]|^{2}\} =&  \rho_{k}^{2}[n]\expect\Big\{|\hat{d}_{kk}|^{2}\Big\} + \rho_{k}^{2}[n]\expect\Big\{|\Tilde{d}_{kk}|^{2}\Big\} + \bar{\rho}_{k}^{2}[n]\expect\Big\{|z_{kk}[n]|^{2}\Big\}.
\end{align}
Therefore, 
\begin{align}
    \expect\Big\{|d_{kk}[n]|^{2}\Big|\hat{d}_{kk}\Big\} =&  \rho_{k}^{2}[n]|\hat{d}_{kk}|^{2} + \rho_{k}^{2}[n]\expect\Big\{|\Tilde{d}_{kk}|^{2}\Big\} + \bar{\rho}_{k}^{2}[n]\expect\Big\{|z_{kk}[n]|^{2}\Big\}.
\end{align}
Substituting the above expression in \eqref{eq:SINR_Denominator}, we obtain
\begin{align}
    \Ed\sum_{k'=1}^{K}\expect\Big\{|d_{kk'}[n]|^{2}\Big|\hat{d}_{kk}\Big\}& - \Ed \big|\expect\Big\{d_{kk}[n]\Big|\hat{d}_{kk}\Big\}\big|^{2} + 1 \nonumber \\
    = \Ed\sum_{k'\neq k}^{K}&\expect\Big\{|d_{kk'}[n]|^{2}\Big|\hat{d}_{kk}\Big\} + \Ed\rho_{k}^{2}[n]\expect\Big\{|\Tilde{d}_{kk}|^{2}\Big\} + \Ed\bar{\rho}_{k}^{2}[n]\expect\Big\{|z_{kk}[n]|^{2}\Big\} + 1. \label{eq:SINR_Denominator_Final}
\end{align}
Substituting \eqref{eq:SINR_Numerator_Final} and \eqref{eq:SINR_Denominator_Final} back in \eqref{eq:capacity0}, we obtain the following expression for the lower bound on the downlink SE:
\begin{equation}\label{eq:capacity1}
    \SE_{k}^{\DT}[n] = \expect\Bigg\{\text{log}_{2}\Bigg(1 + \frac{ \rho_{k}^{2}[n]|\hat{d}_{kk}|^{2}}{\sum_{k'\neq k}^{K}\expect\Big\{|d_{kk'}[n]|^{2}\Big|\hat{d}_{kk}\Big\} + \rho_{k}^{2}[n]\expect\Big\{|\Tilde{d}_{kk}|^{2}\Big\} + \bar{\rho}_{k}^{2}[n]\expect\Big\{|z_{kk}[n]|^{2}\Big\} +  \frac{1}{\Ed}}\Bigg)\Bigg\}.
\end{equation}
To further approximate the lower bound, the outermost expectation in the above expression can be taken inside the logarithm \cite[Lemma 1]{zhang_jsp_2014}. This gives us the following approximation of the achievable downlink SE:
\begin{equation}\label{eq:capacity2}
    \SE_{k}^{\DT}[n] \approx \text{log}_{2}\left(1 + \frac{ \rho_{k}^{2}[n]\expect\Big\{|\hat{d}_{kk}|^{2}\Big\}}{\sum_{k'\neq k}^{K}\expect\Big\{|d_{kk'}[n]|^{2}\Big\} + \rho_{k}^{2}[n]\expect\Big\{|\Tilde{d}_{kk}|^{2}\Big\} + \bar{\rho}_{k}^{2}[n]\expect\Big\{|z_{kk}[n]|^{2}\Big\}  +  \frac{1}{\Ed}}\right).
\end{equation}
Focusing on the numerator term of the effective SINR in the above expression, it can be shown that 
\begin{align}
    \expect\{|\hat{d}_{kk}|^{2}\} =& |\expect\{d_{kk}\}|^{2} + \frac{|\covariance\{d_{kk}, \breve{y}_{\text{dp}, k}\}|^{2}}{|\variance\{\breve{y}_{\text{dp}, k}\}|^{2}}\variance\{\breve{y}_{\text{dp}, k}\} \nonumber \\
    =& \left(\sum_{m=1}^{M}L\sqrt{\eta_{mk}}\gamma_{mk}\right)^{2} + \frac{\left(\sqrt{\Tdp\Edp}\sum_{m=1}^{M}L\eta_{mk}\gamma_{mk}\beta_{mk}\right)^{2}}{1 + \Tdp\Edp\sum_{m=1}^{M}\sum_{k'=1}^{K}\eta_{mk'}\gamma_{mk'}L\beta_{mk}|\bpsi_{k'}^{\hermitian}\bpsi_{k}|^{2}}. \label{eq:TP1_SINR_NUM}
\end{align}
For the denominator of the SINR in \eqref{eq:capacity2}, we compute $\expect\{|d_{kk'}[n]|^{2}\}$ by writing
\begin{align}
    \expect\{|d_{kk'}[n]|^{2}\} =&  \rho_{k}^{2}[n]\sum_{m=1}^{M}\sum_{n=1}^{M}\sqrt{\eta_{mk'}\eta_{nk'}}\expect\{\g_{mk}^{\transpose}\hat{\g}_{mk'}^{*}\g_{nk}^{\hermitian}\hat{\g}_{nk'}\} \nonumber \\
    & + \bar{\rho}_{k}^{2}[n]\sum_{m=1}^{M}\sum_{n=1}^{M}\sqrt{\eta_{mk'}\eta_{nk'}}\expect\{\mathbf{z}_{mk}^{\transpose}\hat{\g}_{mk'}^{*}\mathbf{z}_{nk}^{\hermitian}\hat{\g}_{nk'}\}.
\end{align}
Simplifying the above expression, we obtain
\begin{align}
    \expect\{|d_{kk'}[n]|^{2}\} =&  \rho_{k}^{2}[n]\left(\sum_{m=1}^{M}L\sqrt{\eta_{mk'}} \gamma_{mk'} \frac{\beta_{mk}}{\beta_{mk'}}\right)^{2}|\bphi_{k'}^{\hermitian} \bphi_{k}|^{2}  + \sum_{m=1}^{M}L\eta_{mk'}\gamma_{mk'}\beta_{mk}. \label{eq:TP1_SINR_DEN1}
\end{align}
To compute $\expect\{|\Tilde{d}_{kk}|^{2}\}$ in \eqref{eq:capacity2}, we must evaluate $\expect\{|d_{kk}|^{2}\}$ first. Substituting $k' = k$ and $n = 0$ in \eqref{eq:TP1_SINR_DEN1}, we obtain
\begin{align}
    \expect\{|d_{kk}|^{2}\} =& \sum_{m=1}^{M}L \eta_{mk}\gamma_{mk}\beta_{mk} + \left(\sum_{m=1}^{M}L\sqrt{\eta_{mk}}\gamma_{mk}\right)^{2}.
\end{align}
Now, the quantity $\expect\{|\Tilde{d}_{kk}|^{2}\}$ can be computed as
\begin{align}
    \expect\{|\Tilde{d}_{kk}|^{2}\} =& \expect\{|d_{kk}|^{2}\} - \expect\{|\hat{d}_{kk}|^{2}\} \nonumber \\
    =& \sum_{m=1}^{M}L \eta_{mk}\gamma_{mk}\beta_{mk} - \frac{\Tdp\Edp L^{2}\left(\sum_{m=1}^{M}\eta_{mk}\gamma_{mk}\beta_{mk}\right)^{2}}{1 +  L  \Tdp\Edp \sum_{m=1}^{M} \sum_{k'=1}^{K} \eta_{mk'} \gamma_{mk'} \beta_{mk} |\bpsi_{k'}^{\hermitian}\bpsi_{k}|^{2}} \label{eq:TP1_SINR_DEN2}
\end{align}
Finally, the quantity $\expect\{|z_{kk}[n]|^{2}\}$ in \eqref{eq:capacity2} is found as
\begin{align}
    \expect\{|z_{kk}[n]|^{2}\} =& \sum_{m=1}^{M}\sum_{n=1}^{M} \sqrt{\eta_{mk}\eta_{nk}}\expect\{\mathbf{z}_{mk}^{\transpose}[n]\hat{\g}_{mk}^{*}\mathbf{z}_{nk}^{\hermitian}[n]\hat{\g}_{nk}\} = \sum_{m=1}^{M}L\eta_{mk}\gamma_{mk}\beta_{mk}. \label{eq:TP1_SINR_DEN3}
\end{align}
Substituting \eqref{eq:TP1_SINR_NUM}, \eqref{eq:TP1_SINR_DEN1}, \eqref{eq:TP1_SINR_DEN2} and \eqref{eq:TP1_SINR_DEN3} in \eqref{eq:capacity2}, we obtain the result in proposition 1.

\section{Proof of Proposition 2}
\label{app:proof_proposition_2}

Let $\bphi_{lk} \in \complex^{\Tup} $ denote the uplink pilot assigned to $\UE_{lk}$. The uplink training signal recieved at BS $j$ is
\begin{equation}
    \mathbf{Y}_{\text{up},j} = \sqrt{\Tup\Eup}\sum_{l'=1}^{L_c}\sum_{i'=1}^{K_c}\g_{l'i'}^{j}\bphi_{l'i'}^{\hermitian} + \mathbf{W}_{\text{up},j},
\end{equation}
where $\Eup$ denotes the normalized transmit SNR for the uplink and $\mathbf{W}_{\text{up},m}$ consisting of i.i.d $\CN\left(0,1\right)$ entries denotes noise at the $j$th BS. Now, BS $j$ correlates the received signal with pilot $\bphi_{li}$ to obtain
\begin{align}
    \mathbf{y}_{\text{up},jli} = \mathbf{Y}_{\text{up},j}\bphi_{li} = \sqrt{\Tup\Eup}\sum_{l'=1}^{L_c}\sum_{i'=1}^{K_c}\g_{l'i'}^{j}\bphi_{l'i'}^{\hermitian}\bphi_{li} + \mathbf{W}_{\text{up},j}\bphi_{li}.
\end{align}
Given $\mathbf{y}_{\text{up},jli}$, the MMSE channel estimate $\hat{\mathbf{h}}_{li}^{j}$ is obtained as 
\begin{align}
    \hat{\mathbf{g}}_{li}^{j} = \frac{\sqrt{\Tup\Eup}\beta_{li}^{j}}{\Tup\Eup\sum_{l'=1}^{L_c}\beta_{l'i}^{j} + 1}\mathbf{y}_{\text{up},jli} = c_{li}^{j}\mathbf{y}_{\text{up},jli}.
\end{align}
In the above expression, we have $\hat{\mathbf{g}}_{li}^{j} \sim \CN\left(0, \gamma_{li}^{j}\mathbf{I}_{M_c}\right)$ where $\gamma_{li}^{j} = \sqrt{\Tup\Eup}c_{li}^{j}\beta_{li}^{j}$.

Let $\bpsi_{lk} \in \complex^{\Tdp}$ denote the downlink pilot sequence intended for $\UE_{lk}$. The BS uses conjugate beamforming to transmit downlink pilots. The downlink pilot signal received at $\UE_{lk}$ is
\begin{align}
    \mathbf{y}_{\text{dp},lk} = \sqrt{\Tdp\Edp}\sum_{l'=1}^{L}\sum_{i=1}^{K_c}\sqrt{\eta_{l'i}}\g_{lk}^{l'\transpose}\hat{\mathbf{g}}_{l'i}^{l'*}\bpsi_{l'i}^{\hermitian} + \mathbf{w}_{\text{dp},lk}.
\end{align}
where $\Edp$ denotes the normalized transmit SNR for the downlink, $\eta_{l'i}$ denotes the power control coefficient intended for $\UE_{l'i}$, and $\mathbf{w}_{\textrm{dp},lk} \in \complex^{1 \times \Tdp}$ denotes additive Gaussian noise at $\UE_{lk}$. Next, $\UE_{lk}$ correlates the above received signal with pilot $\bpsi_{lk}$ to obtain
\begin{align}
    \breve{y}_{\text{dp},lk} = \sqrt{\Tdp\Edp} \sqrt{\eta_{lk}} \g_{lk}^{l\transpose}\hat{\mathbf{g}}_{lk}^{l*} + \sqrt{\Tdp\Edp}\sum_{\substack{l'=1 \\ l' \neq l}}^{L_c}\sqrt{\eta_{l'k}}\g_{lk}^{l'\transpose}\hat{\mathbf{g}}_{l'k}^{l'*}\bpsi_{l'k}^{\hermitian}\bpsi_{lk} + \Tilde{w}_{\text{dp},lk},
\end{align}
where $\Tilde{w}_{\text{dp},lk} = \mathbf{w}_{\text{dp},lk}\bpsi_{lk}$. The above expression can be rewritten as
\begin{align}
    \breve{y}_{\text{dp},lk} =& \sqrt{\Tdp\Edp} d_{lkk}^{l} + \sqrt{\Tdp\Edp}\sum_{\substack{l'=1 \\ l' \neq l}}^{L_c}d_{lkk}^{l'} + \Tilde{w}_{\text{dp},lk},
\end{align}
where $d_{lkk}^{l'} = \sqrt{\eta_{l'k}}\g_{lk}^{l'\transpose}\hat{\mathbf{g}}_{l'k}^{l'*}$ represents the downlink channel. The MMSE estimate of $d_{lkk}^{l}$ can be found as \cite{kay_book_1993}
\begin{equation}
    \hat{d}_{lkk}^{l} = M_{c}\sqrt{\eta_{lk}}\gamma_{lk}^{l} + \frac{\sqrt{\Tdp\Edp}M_{c}\eta_{lk}\gamma_{lk}^{l}\beta_{lk}^{l}}{ 1 + \Tdp\Edp\sum_{i=1}^{L_c}M_{c}\eta_{ik}\gamma_{ik}^{i}\beta_{lk}^{i}}\left(\breve{y}_{\textrm{dp}, lk} - \sqrt{\Tdp\Edp}\sum_{i=1}^{L_c}M_{c}\sqrt{\eta_{ik}}\gamma_{ik}^{i}\frac{\beta_{lk}^{i}}{\beta_{ik}^{i}}\right).
\end{equation}
The downlink channel is then given by $d_{lkk}^{l} = \hat{d}_{lkk}^{l} + \Tilde{d}_{lkk}^{l}$ where $\Tilde{d}_{lkk}^{l}$ is the channel estimation error. Next, the $j$th BS proceeds to transmit data symbols on the downlink. The downlink data vector transmitted by BS $j$ at the $n$th-instant ($n = \Tup+\Tdp, \dots, \Tup+\Tdp+\Tdd-1$) is
\begin{align}
    \mathbf{x}_{j}[n] = \sqrt{\Ed}\sum_{i=1}^{K_c}\sqrt{\eta_{ji}}\hat{\mathbf{g}}_{ji}^{j*}q_{ji}[n].
\end{align}
Here, $q_{ji}[n]$ denotes the $n$th data symbol intended for $\UE_{ji}$; the symbols $\{q_{ji}[n]\}$ have zero-mean and unit variance and they are mutually uncorrelated. Now, $\UE_{lk}$ receives the $n$th transmit vector in the form
\begin{align}\label{eq:receive_data_cellular}
    r_{\text{d},lk}[n] = \sqrt{\Ed}d_{lkk}^{l}[n]q_{lk}[n] + \sqrt{\Ed}\sum_{\substack{i=1 \\ i\neq k}}^{K_c}d_{lki}^{l}[n]q_{li}[n] + \sqrt{\Ed}\sum_{\substack{l'=1 \\ l' \neq l}}^{L_c}\sum_{i=1}^{K_c}d_{lki}^{l'}[n]q_{l'i}[n] + w_{\text{d},lk}[n],
\end{align}
where $d_{lki}^{l'}[n] = \sqrt{\eta_{l'i}}\g_{lk}^{l'\transpose}[n]\hat{\mathbf{g}}_{l'i}^{l'*}$ is the $n$th-instant downlink channel and $w_{\text{d},lk}[n]$ is zero-mean, unit-variance noise at $\UE_{lk}$. Using the same approach as in appendix \ref{app:proof_proposition_1}, the approximate lower bound on the $n$th-instant downlink SE of $\UE_{lk}$ can be shown to be
\begin{align}
    \textrm{SE}^{\textrm{cell, DT}}_{lk}[n] = \log_{2}\left(1 + \textrm{SINR}_{lk}^{\textrm{cell,DT}}[n]\right),
\end{align}
where
\begin{align}
    \textrm{SINR}_{lk}^{\textrm{cell,DT}}[n] =& \frac{\left(\textrm{SINR}_{lk}^{\textrm{cell,DT}}[n]\right)_{\textrm{num}}}{\left(\textrm{SINR}_{lk}^{\textrm{cell,DT}}[n]\right)_{\textrm{den}}}
\end{align}
denotes the effective downlink SINR at the $n$th-instant with
\begin{align}
    \left(\textrm{SINR}_{lk}^{\textrm{cell,DT}}[n]\right)_{\textrm{num}} =&  \Ed\rho_{lk}^{2}[n]\,\expect\Big\{|\hat{d}_{lkk}^{l}|^{2}\Big\} \\
    \left(\textrm{SINR}_{lk}^{\textrm{cell,DT}}[n]\right)_{\textrm{den}} =& \Ed\sum_{l'\neq l}^{L_c}\sum_{i=1}^{K_c}\expect\Big\{|d_{lki}^{l'}[n]|^{2}\Big\} + \Ed\sum_{i \neq k}^{K_c}\expect\Big\{|d_{lki}^{l}[n]|^{2}\Big\} +  \Ed\rho_{lk}^{2}[n]\,\expect\Big\{|\Tilde{d}_{lkk}^{l}[n]|^{2}\Big\} \nonumber \\
    & + \Ed\bar{\rho}_{lk}^{2}[n]\,\expect\Big\{|z_{llkk}[n]|^{2}\Big\} + 1,
\end{align}
in which $z_{llkk}[n] = \sqrt{\eta_{lk}}\mathbf{z}_{lk}^{l\transpose}[n]\hat{\mathbf{g}}_{lk}^{l*}$. Upon evaluating the expectations in the above two equations, we obtain the closed-form expression given in proposition 2.

\section{Proof of Proposition 3}
\label{app:proof_proposition_3}
When UEs rely on channel statistics to decode data symbols, the recieved signal in \eqref{eq:receive_data_cellular} can be rewritten as \cite{ngo_twc_2017, marzetta_book}
\begin{align}
    r_{\textrm{d},lk}[n] =& \sqrt{\Ed}\expect\Big\{d_{lkk}^{l}[n]\Big\}q_{lk}[n] + \sqrt{\Ed}\left(d_{lkk}^{l}[n] - \sqrt{\Ed}\expect\Big\{d_{lkk}^{l}[n]\Big\}\right)q_{lk}[n] \nonumber \\
    & + \sqrt{\Ed}\sum_{\substack{i=1 \\ i\neq k}}^{K_c}d_{lki}^{l}[n]q_{li}[n] + \sqrt{\Ed}\sum_{\substack{l'=1 \\ l' \neq l}}^{L_c}\sum_{i=1}^{K_c}d_{lki}^{l'}[n]q_{l'i}[n] + w_{\text{d},lk}[n],
\end{align}
where the first term (containing the mean value of the downlink channel) represents the desired signal term and the rest of the terms form the "effective noise" that is uncorrelated with the desired signal term. A closed-form expression for the lower bound on the $n$th-instant downlink SE can be obtained by taking the ratio of the mean-square value of the desired signal term and the mean-square value of the effective noise to form the effective SINR and then using the hardening bound from \cite{marzetta_book}. This leads us to the result in proposition 3.

\ifCLASSOPTIONcaptionsoff
  \newpage
\fi



\bibliographystyle{IEEEtran}
\bibliography{IEEEabrv.bib, biblio.bib}
\end{document}